\documentclass[twoside,11pt]{article}

\usepackage{times}
\usepackage{amsmath}
\usepackage{amssymb}
\usepackage{graphicx}
\usepackage{stmaryrd}
\usepackage{mathrsfs}

\makeatletter
\def\@maketitle{%
  \newpage
  \null
  \begin{center}%
  \let \footnote \thanks
    {\LARGE\bf \@title \par}
    \vskip 2em
    {\large
      \lineskip .5em%
      \begin{tabular}[t]{c}%
        \@author
      \end{tabular}}%
    \vskip .3em
  \end{center}%
  \par}
\makeatother

\renewenvironment{abstract}{\begin{list}{}
			{\rightmargin\leftmargin
			\listparindent 1.5em
			\parsep 0pt plus 1pt}
			\small\item}{\end{list}}

\newtheorem{defi}{Definition}
\newtheorem{theo}{Theorem}
\newtheorem{prop}{Proposition}
\newtheorem{lemm}{Lemma}
\newtheorem{exam}{Example}
\newtheorem{conj}{Conjecture}
\newtheorem{obs}{Observation}

\newenvironment{definition}[1]{\vspace{-1ex}\begin{defi} \rm \label{#1} }{\vspace{-1ex}\end{defi}}
\newenvironment{theorem}[1]{\vspace{-1ex}\begin{theo} \rm \label{#1} }{\vspace{-1ex}\end{theo}}
\newenvironment{proposition}[1]{\vspace{-1ex}\begin{prop} \rm \label{#1} }{\vspace{-1ex}\end{prop}}
\newenvironment{example}{\vspace{-1ex}\begin{exam} \rm }{\vspace{-1ex}\end{exam}}
\newenvironment{conjecture}[1]{\vspace{-1ex}\begin{conj} \rm \label{#1} }{\vspace{-1ex}\end{conj}}
\newenvironment{observation}[1]{\vspace{-1ex}\begin{obs} \rm \label{#1} }{\vspace{-1ex}\end{obs}}

\newenvironment{itemise}{\begin{list}{$-$}{\leftmargin 18pt
                        \labelwidth\leftmargini\advance\labelwidth-\labelsep
                        \topsep -8pt \itemsep -4pt \parsep 4pt}}{\end{list}}

\newenvironment{definitionG}[2]{\vspace{-1ex}\begin{defi} \rm \label{#1} \mbox{#2}
\begin{itemise} \item[]}{\vspace{-3pt}\end{itemise}\end{defi}}

\newenvironment{theoremG}[2]{\vspace{-1ex}\begin{theo} \rm \label{#1} \mbox{#2}
\begin{itemise} \item[]}{\vspace{-1ex}\end{itemise}\end{theo}}

\newenvironment{lemmaG}[2]{\vspace{-1ex}\begin{lemm} \rm \label{#1} \mbox{#2}
\begin{itemise} \item[]}{\vspace{-1ex}\end{itemise}\end{lemm}}

\newenvironment{lemmai}[2]{\vspace{-1ex}\begin{lemm} \rm \label{lem-#1} \mbox{#2}
\begin{enumerate} \itemsep -5pt \vspace{-1em}}{\vspace{-1ex}\end{enumerate}\end{lemm}}

\newcommand{\refdf}[1]{Definition~\ref{df-#1}}
\newcommand{\refpr}[1]{Proposition~\ref{pr-#1}}
\newcommand{\refthm}[1]{Theorem~\ref{thm-#1}}
\newcommand{\reflem}[1]{Lemma~\ref{lem-#1}}
\newcommand{\reffig}[1]{Figure~\ref{fig-#1}}

\newenvironment{proof}{\begin{trivlist} \item[\hspace{\labelsep}\bf
    Proof\ \ ]}{\vspace{-1ex}\hfill$\boxempty$\end{trivlist}}

\def\text#1{\textrm{#1}}
\def\precond#1{{}^\bullet #1}
\def\postcond#1{{#1}^\bullet}
\def\iprecond#1{{}^\circ#1}
\def\ipostcond#1{{#1}^\circ}
\def\Production#1{\stackrel{#1}{\Longrightarrow}}
\def\production#1{\stackrel{#1}{\longrightarrow}}
\def\equivalent{\Leftrightarrow}
\newfont{\fsc}{eusm10 scaled 1100}      
\def\powerset#1{\mbox{\fsc P}(#1)}
\def\powermultiset#1{\IN^{#1}}
\def\implies{\Rightarrow}

\def\equivalent{\Leftrightarrow}
\def\mathrlap{\mathpalette\mathrlapinternal}
\def\mathrlapinternal#1#2{%
  \rlap{$\mathsurround=0pt#1{#2}$}}
\def\mathllap{\mathpalette\mathllapinternal}
\def\mathllapinternal#1#2{%
  \llap{$\mathsurround=0pt#1{#2}$}}
\def\trail#1{\text{~#1}}
\def\into{\rightarrow}
\def\tb{\tau^\leftarrow}

\def\tbtci{\tau^\Leftarrow}
\def\tbsafetci{\tau^{\Longleftarrow}}
\def\fsidistance{d}
\def\fsivalidmarking{\alpha}
\def\tcidistance{d}
\def\tcivalidmarking{\beta}

\def\refitem#1{``\ref{#1}''}
\def\contradiction{\scalebox{1.4}{$\lightning$}}
\def\defitem#1{\emph{#1}}

\def\rpair#1{\mathord{<}#1\mathord{>}}
\def\FSA{\text{\it FA}}
\def\SA{\text{\it SA}}
\def\AA{\text{\it AA}}

\let\origexists\exists
\let\orignexists\nexists
\let\origforall\forall
\def\exists#1.{\origexists#1.\onespace}
\def\nexists#1.{\orignexists#1.\onespace}
\def\forall#1.{\origforall#1.\onespace}
\def\onespace#1{\let\argument=#1\ifx\onespace#1\else~\fi\argument}

\def\readyset{\mathscr{R}}
\def\structuralN{\mbox{\sf N}}
\def\structuralM{\mbox{\sf M}}
\let\origmin\min
\def\min{\mathord{\origmin}}
\let\origmax\max
\def\max{\mathord{\origmax}}
\def\circled#1{\pscirclebox[linewidth=0.015,framesep=2pt]{#1}}
\def\boxed#1{\,\psframebox[linewidth=0.015,framesep=2pt]{#1}\,}
\def\quireunderscore{_}
\def\quire#1{%
  \def\tmp{#1}%
  \ifx\tmp\quireunderscore%
    \def\tmp{\quireindexed_}
  \else%
    \def\tmp{\mathscr{Q}#1}
  \fi\tmp}
\def\quireindexed_#1{\mathscr{Q}_{\text{#1}}}

\makeatletter
\def\goesto{\@transition\rightarrowfill}
\def\Goesto{\@transition\Rightarrowfill}
\def\ngoesto{\@transition\nrightarrowfill}
\def\nGoesto{\@transition\nRightarrowfill}
\def\@transition#1{\@ifnextchar[{\@@transition{#1}}{\@@transition{#1}[]}}
\newbox\@transbox
\newbox\@arrowbox
\def\rightarrowfill{$\m@th\mathord-\mkern-6mu%
  \cleaders\hbox{$\mkern-2mu\mathord-\mkern-2mu$}\hfill
  \mkern-6mu\mathord\rightarrow$}
\def\Rightarrowfill{$\m@th\mathord=\mkern-6mu%
  \cleaders\hbox{$\mkern-2mu\mathord=\mkern-2mu$}\hfill
  \mkern-6mu\mathord\Rightarrow$}
\def\@@transition#1[#2]%
   {\setbox\@transbox\hbox
       {\vrule height 1.5ex depth 1ex width 0ex\hskip0.25em$\scriptstyle#2$\hskip0.25em}
   \ifdim\wd\@transbox<1.5em
      \setbox\@transbox\hbox to 1.5em{\hfil\box\@transbox\hfil}\fi
   \setbox\@arrowbox\hbox to \wd\@transbox{#1}
   \ht\@arrowbox\z@\dp\@arrowbox\z@
   \setbox\@transbox\hbox{$\mathop{\box\@arrowbox}\limits^{\box\@transbox}$}
   \ht\@transbox\z@\dp\@transbox\z@
   \mathrel{\box\@transbox}}
\makeatother

\def\alignedcaption[#1&#2]{\mbox{\scriptsize $\mathllap{#1{}}\mathrlap{#2}$}}

\def\ie{i.e.\ }
\def\varnothing{\emptyset}
\def\Act{\textrm{\upshape Act}}
\def\Loc{\textrm{\upshape Loc}}
\def\concurrent{\smile}
\def\conflict{\mathrel{\#}}
\def\restrictedto{\upharpoonright}

\newcommand{\IN}{\mbox{\rm I\hspace{-1.5pt}N}}        
\newcommand{\dcup}{\stackrel{\mbox{\huge .}}{\cup}}   
\newcommand{\plat}[1]{\raisebox{0pt}[0pt][0pt]{#1}}   
\newcommand{\inp}{\mathbin\in}                        
\newcommand{\notinp}{\mathbin{\notin}}                

\usepackage{pstricks}
\usepackage{pst-node}
\usepackage{graphicx}

\newcounter{netimage}
\def\p#1:#2;{\cnode #1{0.3}{n\thenetimage-#2}}
\def\P#1:#2;{\p #1:#2;\pscircle*#1{0.1}}
\def\q#1:#2:#3;{\p #1:#2;\rput#1{\rput[l](0.45,0){\large\it #3}}}
\def\Q#1:#2:#3;{\P #1:#2;\rput#1{\rput[l](0.45,0){\large\it #3}}}
\def\t#1:#2:#3;{\rput#1{\rnode{n\thenetimage-#2}{\psframebox{%
  \vbox to 0.6cm{\vfil\hbox to 0.6cm{\hfil\Large\it #3\hfil}\vfil}}}}}
\def\u#1:#2:#3:#4;{\rput#1{\rnode{n\thenetimage-#2}{\psframebox{%
  \vbox to 0.6cm{\vfil\hbox to 0.6cm{\hfil\Large\it #3\hfil}\vfil}}}}%
  \rput#1{\rput[l](0.6,0){\large\it #4}}}
\def\a#1->#2;{\ncline{->}{n\thenetimage-#1}{n\thenetimage-#2}}
\def\A#1->#2;{\ncarc{->}{n\thenetimage-#1}{n\thenetimage-#2}}
\def\av#1[#2]-#3->[#4]#5;{
  \SpecialCoor
  \psline[linearc=0.2]{->}([angle=#2]n\thenetimage-#1)#3([angle=#4]n\thenetimage-#5)
}

\makeatletter
\long\def\petrinet(#1)#2\end{\psscalebox{0.7}{\pspicture(#1)\stepcounter{netimage}#2\endpspicture}\end}

\makeatother

\def\lastname{van Glabbeek, Goltz and Schicke}
\def\titleheader{On Synchronous and Asynchronous Interaction in Distributed Systems}
\title{On Synchronous and Asynchronous Interaction in Distributed Systems}
\author{Rob van Glabbeek\\
 \footnotesize NICTA, Sydney, Australia\\[-3pt]
 \footnotesize University of New South Wales, Sydney, Australia\\
 \footnotesize \tt rvg@cs.stanford.edu
\and\hspace{-1em}
    Ursula Goltz
\hspace{4em}
  Jens-Wolfhard Schicke\thanks{This paper was partially
  written during a four month stay of J.-W. Schicke at NICTA,
  during which he was supported by DAAD (Deutscher Akademischer Austauschdienst) and NICTA.
}\\
 \footnotesize Institute for Programming and Reactive Systems\\[-3pt]
 \footnotesize TU Braunschweig, Braunschweig, Germany\\
 \footnotesize \tt 
\hspace{-4.5em}goltz@ips.cs.tu-bs.de
\hspace{3.5em} drahflow@gmx.de}

\begin{document}
\maketitle
\thispagestyle{empty}

\begin{abstract}
   When considering distributed systems, it is a central issue how to
   deal with interactions between components. In this paper, we
   investigate the paradigms of synchronous and asynchronous
   interaction in the context of distributed systems.  We investigate
   to what extent or under which conditions synchronous interaction is
   a valid concept for specification and implementation of such
   systems.  We choose Petri nets as our system model and consider
   different notions of distribution by associating locations to
   elements of nets.  First, we investigate the concept of
   simultaneity which is inherent in the semantics of Petri nets when
   transitions have multiple input places. We assume that tokens may
   only be taken instantaneously by transitions on the same
   location. We exhibit a hierarchy of `asynchronous' Petri net
   classes by different assumptions on possible distributions.
   Alternatively, we assume that the synchronisations specified in a
   Petri net are crucial system properties. Hence transitions and
   their preplaces may no longer placed on separate locations. We then
   answer the question which systems may be implemented in a
   distributed way without restricting concurrency, assuming that
   locations are inherently sequential.  It turns out that in both
   settings we find semi-structural properties of Petri nets
   describing exactly the problematic situations for interactions in
   distributed systems.
\end{abstract}

\section{Introduction}

In this paper, we address interaction patterns in distributed
systems. By a distributed system we understand here a system which is
executed on spatially distributed locations, which do not share a
common clock (for performance reasons for example). We want to
investigate to what extent or under which conditions synchronous
interaction is a valid concept for specification and implementation of
such systems. It is for example a well-known fact that synchronous
communication can be simulated by asynchronous communication using
suitable protocols.  However, the question is whether and under which
circumstances these protocols fully retain the original behaviour of a
system.  What we are interested in here are precise descriptions of
what behaviours can possibly be preserved and which cannot.
 
The topic considered here is by no means a new one. We give a short
overview on related approaches in the following.

Already in the 80th, Luc Boug\'e considered a similar problem in the
context of distributed algorithms. In \cite{bouge88symmetricleader} he
considers the problem of implementing symmetric leader election in the
sublanguages of CSP obtained by allowing different forms of
communication, combining input and output guards in guarded choice in
different ways. He finds that the possibility of implementing leader
election depends heavily on the structure of the communication
graphs. Truly symmetric schemes are only possible in CSP with
arbitrary input and output guards in choices.

Synchronous interaction is a basic concept in many languages for
system specification and design, e.g. in statechart-based approaches,
in process algebras or the $\pi$-calculus. For process algebras and
the $\pi$-calculus, language hierarchies have been established which
exhibit the expressive power of different forms of synchronous and
asynchronous interaction.  In \cite{boer91embedding} Frank de Boer and
Catuscia Palamidessi consider various dialects of CSP with differing
degrees of asynchrony.  Similar work is done for the $\pi$-calculus in
\cite{palamidessi97comparing} by Catuscia Palamidessi, in
\cite{nestmann00what} by Uwe Nestmann and in \cite{G:FoSSaCS06} by
Dianele Gorla.  A rich hierarchy of asynchronous $\pi$-calculi has
been mapped out in these papers. Again mixed-choice, i.e. the ability
to combine input and output guards in a single choice, plays a central
r\^ ole in the implementation of truly synchronous behaviour.

In \cite{selinger97firstorder}, Peter Selinger considers labelled
transition systems whose visible actions are partitioned into input and
output actions. He defines asynchronous implementations of such a
system by composing it with in- and output queues, and then
characterises the systems that are behaviourally equivalent to their
asynchronous implementations. The main difference with our approach is
that we focus on asynchrony within a system, whereas Selinger focusses
on the asynchronous nature of the communications of a system with the
outside world.

Also in hardware design it is an intriguing quest to use interaction
mechanisms which do not rely on a global clock, in order to gain
performance. Here the simulation of synchrony by asynchrony can be a
crucial issue, see for instance \cite{lamport78ordering} and
\cite{lamport02arbitration}.

In contrast to the approaches based on language constructs like the
work on CSP or the $\pi$-calculus, we choose here a very basic system
model for our investigations, namely Petri nets. The main reason for
this choice is the detailed way in which a Petri net represents a
concurrent system, including the interaction between the components it
may consist of. In an interleaving based model of concurrency such as
labelled transition systems modulo bisimulation semantics, a system
representation as such cannot be said to contain synchronous or
asynchronous interaction; at best these are properties of composition
operators, or communication primitives, defined in terms of such a
model. A Petri net on the other hand displays enough detail of a
concurrent system to make the presence of synchronous communication
discernible. This makes it possible to study synchronous and
asynchronous interaction without digressing to the realm of
composition operators.
  
  Also in Petri net theory, the topic which concerns us here has
  already been tackled. It has been investigated in
  \cite{hopkins91distnets} and \cite{taubner88zurverteiltenimpl}
  whether and how a Petri net can be implemented in a distributed
  way. We will comment on these and other related papers in the area
  of Petri net theory in the conclusion.
 
  In a Petri net, a transition interacts with its preplaces by
  consuming tokens. In Petri net semantics, taking a token is usually
  considered as an instantaneous action, hence a synchronous
  interaction between a transition and its preplace. In particular
  when a transition has several preplaces this becomes a crucial
  issue. In this paper we investigate what happens if we consider a
  Petri net as a specification of a system that is to be implemented
  in a distributed way. For this we introduce locations on which all
  elements of a Petri net have to be placed upon. The basic assumption
  is that interaction between remote components takes time.  In our
  framework this means that the removal of a token will be considered
  instantaneous only if the removing transition and the place where
  the token is removed from are co-located. Our investigations are now
  twofold.

  In Section~\ref{asynchronous} of this paper, we consider under which
  circumstances the synchronous interaction between a transition and
  its preplace may be mimicked asynchronously, thus allowing to put
  places and their posttransitions on different locations. Following
  \cite{glabbeek08symmasymm}, we model the asynchronous interaction
  between transitions and their preplaces by inserting silent
  (unobservable) transitions between them.  We investigate the effect
  of this transformation by comparing the behaviours of nets before
  and after insertion of the silent transitions using a suitable
  equivalence notion.  We believe that most of our results are
  independent of the precise choice of this equivalence. However, as
  explained in Section~\ref{distributable}, it has to preserve
  causality, branching time and divergence to some small extent, and
  needs to abstract from silent transitions.  Therefore we choose one
  such equivalence, based on its technical convenience in establishing
  our results. Our choice is \emph{step readiness equivalence}. It is
  a variant of the \emph{readiness equivalence} of \cite{OH86},
  obtained by collecting the set of \emph{steps} of multiple actions
  possible after a certain sequence of actions, instead of just the
  set of possible actions.
  We call a net \textit{asynchronous} if, for a suitable placement of
  its places and transitions, the above-mentioned
  transformation replacing synchronous by asynchronous interaction
  preserves step readiness equivalence.
  Depending on the allowed placements, we obtain
  a hierarchy of
  classes of asynchronous nets:
  \textit{fully asynchronous} nets, \textit{symmetrically
  asynchronous} nets and \textit{asymmetrically asynchronous} nets.
  We give semi-structural properties that characterise precisely when\linebreak[3]
  a net falls into one of these classes. This puts the results from
  \cite{glabbeek08symmasymm} in a uniform framework and extends them by
  introducing a simpler notion of asymmetric asynchrony.
  
 In Sections~\ref{distributed systems} and~\ref{distributable} we
 pursue an alternative approach. We assume that the synchronisations
 specified in a Petri net are crucial system properties.  Hence we
 enforce co-locality between a transition and all its preplaces while
 at the same time assuming that concurrent activity is not possible at
 a single location. We call nets fulfilling these requirement
 \textit{distributed} and investigate which behaviours can be
 implemented by distributed nets.  Again we compare the behaviours up
 to step readiness equivalence.  We call a net \textit{distributable}
 iff its behaviour can be equivalently produced by a distributed
 net. We give a behavioural and a semi-structural characterisation of
 a class of non-distributable nets, thereby exhibiting behaviours
 which cannot be implemented in a distributed way at all.  Finally, we
 give a lower bound of distributability by providing a concrete
 distributed implementation for a wide range of nets.

An extended abstract of this paper will appear in the proceedings of the
33rd International Symposium on {\sl Mathematical Foundations of Computer
Science} (MFCS 2008), Toru\'n, Poland, August 2008 (E. Ochma\'nski \&
J. Tyszkiewicz, eds.), LNCS 5162, Springer, 2008, pp. 16-35.

\section{Basic Notions}
\label{basic}

We consider here 1-safe net systems, \ie places never carry more than
one token, but a transition can fire even if pre- and postset intersect.

\begin{definitionG}{df-nst}
  {Let \Act{} be a set of \emph{visible actions} and
  $\tau\mathbin{\not\in}\Act$ be an \emph{invisible action}.}
  A \defitem{labelled net} (over \Act) is a tuple
  $N = (S, T, F, M_0, \ell)$ where
  \begin{itemise}
    \item $S$ is a set (of \defitem{places}),
    \item $T$ is a set (of \defitem{transitions}),
    \item $F \subseteq
       S \times T \cup T \times S$
      (the \defitem{flow relation}),
    \item $M_0 \subseteq S$ (the \defitem{initial marking}) and
    \item \plat{$\ell: T \into \Act \dcup \{\tau\}$} (the \defitem{labelling function}).
\vspace{5pt}
  \end{itemise}
\end{definitionG}

\noindent
Petri nets are depicted by drawing the places as circles, the
transitions as boxes containing the respective label, and the flow
relation as arrows (\defitem{arcs}) between them.  When a Petri net
represents a concurrent system, a global state of such a system is
given as a \defitem{marking}, a set of places, the initial state being
$M_0$.  A marking is depicted by placing a dot (\defitem{token}) in
each of its places.  The dynamic behaviour of the represented system
is defined by describing the possible moves between markings. A
marking $M$ may evolve into a marking $M'$ when a nonempty set of transitions
$G$ \defitem{fires}. In that case, for each arc $(s,t) \in F$ leading
to a transition $t$ in $G$, a token moves along that arc from $s$ to
$t$.  Naturally, this can happen only if all these tokens are
available in $M$ in the first place. These tokens are consumed by the
firing, but also new tokens are created, namely one for every outgoing
arc of a transition in $G$. These end up in the places at the end of
those arcs.  A problem occurs when as a result of firing $G$ multiple
tokens end up in the same place. In that case $M'$ would not be
a marking as defined above. In this paper we restrict
attention to nets in which this never happens. Such nets are called
\defitem{1-safe}.  Unfortunately, in order to formally define this
class of nets, we first need to correctly define the firing rule
without assuming 1-safety. Below we do this by forbidding the firing
of sets of transitions when this might put multiple tokens in the same
place.

\begin{definitionG}{df-steps}
{Let $N = (S, T, F, M_0, \ell)$ be a labelled net. Let $M_1, M_2 \subseteq S$.}
We denote the preset and postset of a net element $x\in S\cup T$ by
$\precond{x} := \{y \mid (y, x) \in F\}$
and
$\postcond{x} := \{y \mid (x, y) \in F\}$
respectively.
These functions are extended to sets in the usual manner, \ie
$\precond{X} := \{y \mid y \inp \precond{x},~ x \inp X\}$.

A nonempty set of transitions $\varnothing \mathbin{\not=}G \mathbin{\subseteq} T$,
  is called a \defitem{step from $M_1$ to $M_2$},
  notation $M_1 \,[G\rangle_N\, M_2$, if
  \begin{itemise}
    \item all transitions contained in $G$ are \defitem{enabled}, that is
      \vspace{-1ex}\begin{equation*}
        \forall t\in G. \precond{t} \subseteq M_1 \wedge 
          (M_1 \setminus \precond{t}) \cap \postcond{t} = \varnothing \trail{,}
      \vspace{-1ex}\end{equation*}
    \item all transitions of $G$ are \defitem{independent}, that is \defitem{not conflicting}:
      \vspace{-1ex}\begin{equation*}
        \forall t,u \in G, t\not= u. \precond{t} \cap \precond{u} = \varnothing
        \wedge \postcond{t} \cap \postcond{u} = \varnothing \trail{,}
      \vspace{-1ex}\end{equation*}
    \item in $M_2$ all tokens have been removed from the \defitem{preplaces}
      of $G$ and new tokens have been inserted at the \defitem{postplaces} of $G$:
      \vspace{-1ex}\begin{equation*}
        M_2 = \left(M_1 \setminus \precond{G}\right) \cup \postcond{G} \trail{.}
      \end{equation*}
  \end{itemise}
\end{definitionG}

\noindent
To simplify statements about possible behaviours of nets, we use some abbreviations.

\begin{definitionG}{df-steprel}
  {Let $N = (S, T, F, M_0, \ell)$ be a labelled net.}
      We extend the labelling function $\ell$ to (multi)sets elementwise.

      \plat{$\mathord{\production{}_N} \subseteq \powerset{S} \times
      \powermultiset{\Act} \times \powerset{S}$}
      is given by $M_1 \production{A}_N M_2 \equivalent
      \exists\, G \subseteq T. M_1~[G\rangle_N~ M_2 \wedge A = \ell(G)$

      $\mathord{\production{\tau}_N} \subseteq \powerset{S} \times \powerset{S}$
      is defined by $M_1 \production{\tau}_N M_2 \equivalent
      \exists t \inp T. \ell(t) \mathbin= \tau \wedge M_1 ~[\{t\}\rangle_N~ M_2$

      $\mathord{\Production{}_N} \subseteq \powerset{S} \times \Act^* \times
      \powerset{S}$ is defined by
      $M_1 \Goesto[\,a_1 a_2 \cdots a_n~]_N M_2 \equivalent\\[3pt]
      \hphantom{M_1 \Production{\sigma}_N M_2 \equivalent }
      M_1
      \production{\tau}^*_N \production{\{a_1\}}_N
      \production{\tau}^*_N \production{\{a_2\}}_N
      \production{\tau}^*_N \cdots
      \production{\tau}^*_N \production{\{a_n\}}_N
      \production{\tau}^*_N
      M_2$\\[3pt]
      where $\production{\tau}^*_N$ denotes the reflexive and
      transitive closure of $\production{\tau}_N$.

  We write $M_1 \production{A}_N$ for $\exists M_2. M_1
  \production{A}_N M_2$, $M_1 \arrownot\production{A}_N$ for $\nexists M_2. M_1
  \production{A}_N M_2$ and similar for the other two relations.
  Likewise $M_1 [G\rangle_N$ abbreviates $\exists M_2. M_1 [G\rangle_N M_2$.

  A marking $M_1$ is said to be \defitem{reachable} iff there is a
  $\sigma \in \Act^*$ such that $M_0 \Production{\sigma}_N M_1$. The set of all
  reachable markings is denoted by $[M_0\rangle_N$.
\end{definitionG}

\noindent
We omit the subscript $N$ if clear from context.

As said before, here we only want to consider 1-safe nets. Formally,
we restrict ourselves to \defitem{contact-free nets}, where in every
reachable marking $M_1 \in [M_0\rangle$ for all $t \in T$ with
  $\precond{t} \subseteq M_1$\vspace{-1ex}
\begin{equation*}
  (M_1 \setminus \precond{t}) \cap \postcond{t} = \varnothing \trail{.}
\vspace{1pt}
\end{equation*}
For such nets, in \refdf{steps} we can just as well consider a
transition $t$ to be enabled in $M$ iff $\precond{t}\subseteq M$, and
two transitions to be independent when $\precond{t} \cap \precond{u} =
\varnothing$.

In this paper we furthermore restrict attention to nets for which
$\precond{t}\neq\emptyset$ and $\precond{t}$ and $\postcond{t}$ are
finite for all $t\inp T$ and $\postcond{s}$ is finite for all $s \inp
S$.  We also require the initial marking $M_0$ to
be finite.  A consequence of these restrictions is that all reachable
markings are finite, and it can never happen that infinitely many
independent transitions are enabled.
Henceforth, with \emph{net} we mean a labelled net obeying the above
restrictions.

In our nets transitions are labelled with \emph{actions} drawn from a
set \plat{$\Act \dcup \{\tau\}$}. This makes it possible to see these
nets as models of \defitem{reactive systems}, that interact with their
environment. A transition $t$ can be thought of as the occurrence of
the action $\ell(t)$. If $\ell(t)\inp\Act$, this occurrence can be
observed and influenced by the environment, but if $\ell(t)\mathbin=\tau$,
$t$ is an \defitem{internal} or \defitem{silent} transition whose
occurrence cannot be observed or influenced by the environment. Two
transitions whose occurrences cannot be distinguished by the
environment are equipped with the same label. In particular, given
that the environment cannot observe the occurrence of internal
transitions at all, all of them have the same label, namely $\tau$.

We use the term \defitem{plain nets} for nets where $\ell$ is
injective and no transition has the label $\tau$, \ie essentially
unlabelled nets. Similarly, we speak of \defitem{plain $\tau$-nets} to
describe nets where $\ell(t) = \ell(u) \ne \tau \implies t = u$,
\ie nets where every observable action is produced by a unique
transition.
In this paper we focus on plain nets, and give semi-structural
characterisations of classes of plain nets only.
However, in defining whether a net belongs to one of those classes, we
study its implementations, which typically are plain $\tau$-nets.
When proving our impossibility result (\refthm{trulysyngltfullm} in
Section~\ref{distributable}) we even allow arbitrary nets as implementations.

We use the following variation of readiness semantics \cite{OH86} to
compare the behaviour of nets.

\begin{definitionG}{df-readypair}
  {Let $N = (S, T, F, M_0, \ell)$ be a net, $\sigma \in \Act^*$ and
  $X \subseteq \powermultiset{\Act}$.}
  $\rpair{\sigma, X}$ is a \defitem{step ready pair} of $N$ iff\vspace{-3pt}
  $$\exists M. M_0 \Production{\sigma} M \wedge M \arrownot\production{\tau}
  \wedge \, X = \{A \inp \powermultiset{\Act} \mid M \production{A}\}.$$
  We write $\readyset(N)$ for the set of all step ready pairs of $N$.
\\
  Two nets $N$ and $N'$ are \defitem{step readiness equivalent},
  $N \approx_\mathscr{R} N'$, iff $\readyset(N) = \readyset(N')$.
\end{definitionG}
The elements of a set $X$ as above are multisets of actions,
but as in all such multisets that will be mentioned in this paper
the multiplicity of each action occurrence is at most~1, we use set
notation to denote them.

\section{Asynchronous Petri Net Classes}
\label{asynchronous}

  In Petri nets, an inherent concept of simultaneity is built in, since
  when a transition has more than one preplace, it can be crucial that tokens
  are removed instantaneously.  When using a Petri net to model a system which
  is intended to be
  implemented in a distributed way, this built-in concept of
  synchronous interaction may be problematic.

  In this paper, a given net is regarded as a \emph{specification} of
  how a system should behave, and this specification involves complete
  synchronisation of the firing of a transition and the removal of all
  tokens from its preplaces. In this section, we propose various
  definitions of an \emph{asynchronous implementation} of a net $N$,
  in which such synchronous interaction is wholly or partially ruled
  out and replaced by asynchronous interaction. The question to be
  clarified is whether such an asynchronous implementation faithfully
  mimics the dynamic behaviour of $N$. If this is the case, we call
  the net $N$ \emph{asynchronous} with respect to the chosen
  interaction pattern.

  The above programme, and thus the resulting concept of
  asynchrony, is parametrised by the answers to three questions:\vspace{-3pt}
\begin{enumerate}
\vspace{-7pt}
\item Which synchronous interactions do we want to rule out exactly?
\vspace{-7pt}
\item How do we replace synchronous by asynchronous interaction?
\vspace{-7pt}
\item When does one net faithfully mimic the dynamic behaviour of another?
\vspace{-7pt}
\end{enumerate}
  To answer the first question we associate a \emph{location} to each
  place and each transition in a net. A transition may take a token
  instantaneously from a preplace (when firing) iff this preplace is
  co-located with the transition; if the preplace resides on a
  different location than the transition, we have to assume the
  collection of the token takes time, and thus the place looses its
  token \emph{before} the transition fires.

  We model the association of locations to the places and transitions
  in a net $N=(S,T,F,M_0,\ell)$ as a function $D: S\cup T \rightarrow
  \Loc$, with $\Loc$ a set of possible locations.  We refer to such a
  function as a \defitem{distribution} of $N$.  Since the identity of
  the locations is irrelevant for our purposes, we can just as well
  abstract from $\Loc$ and represent $D$ by the equivalence relation
  $\equiv_D$ on $S\cup T$ given by $x \equiv_D y$ iff $D(x)=D(y)$.
  
  In this paper we do not deal with nets that have a distribution
  built in. We characterise the interaction patterns we are interested
  in by imposing particular restrictions on the allowed distributions.
  The implementor of a net can choose any distribution that satisfies
  the chosen requirements, and we call a net asynchronous for a
  certain interaction pattern if it has a correct asynchronous
  implementation based on any distribution satisfying the respective
  requirements.

  The \defitem{fully asynchronous} interaction pattern is obtained by requiring
  that all places and all transitions reside on different locations.
  This makes it necessary to implement the removal of every token in a
  time-consuming way. However, this leads to a rather small class of
  asynchronous nets, that falls short for many applications.  We
  therefore propose two ways to loosen this requirement, thereby
  building a hierarchy of classes of asynchronous nets. Both require
  that all places reside on different locations, but a transition may
  be co-located with one of its preplaces.
  The \defitem{symmetrically asynchronous} interaction pattern allows
  this only for transitions with a single preplace, whereas
  in the \defitem{asymmetrically asynchronous} interaction pattern
  any transition may be
  co-located with one of its preplaces. Since two preplaces can never
  be co-located, this breaks the symmetry between the preplaces of a
  transition; an implementor of a net has to choose at most one
  preplace for every transition, and co-locate the transition with it.
  The removal of tokens from all other preplaces needs to be
  implemented in a time-consuming way. Note that all three interaction
  patterns break the synchronisation of the token removal between
  the various preplaces.

\begin{definitionG}{df-icerequirements}
 {Let $D$ be a distribution on a net $N = (S, T, F, M_0, \ell)$,}
 and let $\equiv_D$ be the induced equivalence relation on $S \cup T$. We say
 that $D$ is
  \begin{itemise}
    \item \defitem{fully distributed},
      $D \in \quire_{FD}$, when
      $x \equiv_D y$ for $x,y \in S\cup T$ only if $x = y$,
    \item \defitem{symmetrically distributed},
      $D \in \quire_{SD}$, when
      \begin{center}
	\begin{tabular}{rll}
	  $p \equiv_D q$ &for $p,q \in S$ &only if $p = q$,\\
	  $t \equiv_D p$ &for $t \inp T\!$, $\,p \inp S$ &only if
	    $\precond{t} = \{p\}$ and\\
	  $t \equiv_D u$ &for $t,u \in T$ &only if
	    $t = u$ or $\exists p \inp S. t \equiv_D p \equiv_D u$,\\
	\end{tabular}
      \end{center}
    \item \defitem{asymmetrically distributed},
      $D \in \quire_{AD}$, when
      \begin{center}
	\begin{tabular}{rll}
	  $p \equiv_D q$ &for $p,q \in S$ &only if $p = q$,\\
	  $t \equiv_D p$ &for $t \inp T\!$, $\,p \inp S$ &only if
	    $p \in \precond{t}$ and\\
	  $t \equiv_D u$ &for $t,u \in T$ &only if
	    $t = u$ or $\exists p \inp S. t \equiv_D p \equiv_D u$.\\
	\end{tabular}
      \end{center}
  \end{itemise}
\end{definitionG}

\noindent
  The second question raised above was: How do we replace synchronous by asynchronous interaction? In this section we assume that if
  an arc goes from a place $s$ to a transition $t$ at a different
  location, a token takes time to move from $s$ to $t$. Formally, we describe this by inserting silent (unobservable)
  transitions between transitions and their remote preplaces.
  This leads to the following notion of an asynchronous implementation of a net
  with respect to a chosen distribution.

\begin{definitionG}{df-fsi}
  {Let $N = (S, T, F, M_0, \ell)$ be a net, and let $\equiv_D$ be an
  equivalence relation on $S\cup T$.}
  The \defitem{$D$-based asynchronous implementation} of $N$ is
  $I_D(N) := (S \cup S^\tau, T\cup T^\tau, F', M_0, \ell')$ with\vspace{-1ex}
  $$\begin{array}{lll}
    S^\tau :=& \{s_t \mid t \in T,~ s \in \precond{t},~ s \not\equiv_D t\} \trail{,}\\[3pt]
    T^\tau :=& \{\mathrlap{t_s}\hphantom{s_t} \mid t \in T,~ s \in \precond{t},~ s \not\equiv_D t\} \trail{,}\\[3pt]
    F' :=& \{(t, s) \mid t \in T,~ s \in \postcond{t}\}
      \cup \{(s, t) &\mid t \in T,~ s \in \precond{t},~ s \equiv_D t\}\\
    & \hfill {} \cup \{(s, t_s),(t_s, s_t),(s_t,t) &\mid t \in T,~
     s \in \precond{t},~ s \not\equiv_D t\} \trail{,}\\[3pt]
    \ell'\restrictedto T &= \ell \qquad \mbox{and} \qquad \ell'(t_s)=\tau 
    \hfill\mbox{for}&~~ t_s \in T^\tau .
  \end{array}$$
\vspace{-2em}
\end{definitionG}
  
\begin{proposition}{pr-}
  For any (contact-free) net $N$, and any choice of $\equiv_D$, the
  net $I_D(N)$ is contact-free, and satisfies the other requirements
  imposed on nets, listed in Section~\ref{basic}.
\end{proposition}

\begin{proof}
  In Appendix~\ref{app-asynchronous}.
\end{proof}

\noindent
  The above protocol for replacing synchronous by asynchronous
  interaction appears to be one of the simplest ones imaginable.
  More intricate protocols, involving many asynchronous messages
  between a transition and its preplaces, could be contemplated, but
  we will not study them here. Our protocol involves just one such
  message, namely from the preplace to its posttransition.
  It is illustrated in \reffig{impl-result}.

\begin{figure}
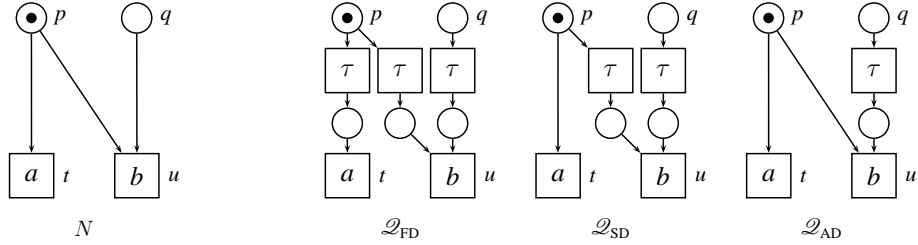

  \begin{center}
    \begin{petrinet}(16,4.5)
      \rput(1,0){\large$N$}

      \Q (0,4):p1:p;
      \q (2,4):p2:q;
      \u (0,1):t1:a:t;
      \u (2,1):t2:b:u;
      \a p1->t1;
      \a p1->t2;
      \a p2->t2;

      \rput(7,0){\large$\quire_{FD}$}

      \Q (6,4):p1b:p;
      \t (6,3):p1btt:$\tau$;
      \p (6,2):p1bt1p;
      \q (8,4):p2b:q;
      \t (8,3):p2btt:$\tau$;
      \p (8,2):p2btp;
      \t (7,3):p1bt2t:$\tau$;
      \p (7,2):p1bt2p;
      \u (6,1):t1b:a:t;
      \u (8,1):t2b:b:u;
      \a p1b->p1btt; \a p1btt->p1bt1p; \a p1bt1p->t1b;
      \a p2b->p2btt; \a p2btt->p2btp; \a p2btp->t2b;
      \a p1b->p1bt2t; \a p1bt2t->p1bt2p; \a p1bt2p->t2b;

      \rput(11,0){\large$\quire_{SD}$}

      \Q (10,4):p1b:p;
      \q (12,4):p2b:q;
      \t (12,3):p2btt:$\tau$;
      \p (12,2):p2btp;
      \t (11,3):p1bt2t:$\tau$;
      \p (11,2):p1bt2p;
      \u (10,1):t1b:a:t;
      \u (12,1):t2b:b:u;
      \a p1b->t1b;
      \a p2b->p2btt; \a p2btt->p2btp; \a p2btp->t2b;
      \a p1b->p1bt2t; \a p1bt2t->p1bt2p; \a p1bt2p->t2b;

      \rput(15,0){\large$\quire_{AD}$}

      \Q (14,4):p1b:p;
      \q (16,4):p2b:q;
      \t (16,3):p2btt:$\tau$;
      \p (16,2):p2btp;
      \u (14,1):t1b:a:t;
      \u (16,1):t2b:b:u;
      \a p1b->t1b;
      \a p2b->p2btt; \a p2btt->p2btp; \a p2btp->t2b;
      \a p1b->t2b;
    \end{petrinet}
  \end{center}
\vspace{-1em}
  \caption{Possible results for $I_D(N)$ given different requirements}
  \label{fig-impl-result}
\end{figure}

\noindent
  The last question above was:
  When does one net faithfully mimic the dynamic behaviour of another?
  This asks for a \emph{semantic equivalence} on Petri
  nets, telling when two nets display the same behaviour.  Many such
  equivalences have been studied in the literature.  We believe that
  most of our results are independent of the precise choice of a
  semantic equivalence, as long as it preserves causality and
  branching time to some degree, and abstracts from silent
  transitions. Therefore we choose one such equivalence, based on its
  technical convenience in establishing our results, and postpone
  questions on the effect of varying this equivalence for further
  research. Our choice is \emph{step readiness equivalence}, as
  defined in Section~\ref{basic}. Using this equivalence, we define a
  notion of \emph{behavioural asynchrony} by asking whether the
  asynchronous implementation of a net preserves its behaviour. This
  notion is parametrised by the chosen interaction pattern,
  characterised as a requirement on the allowed distributions.

\begin{definitionG}{df-behaviourasync}
  {Let $\quire$ be a requirement on distributions of nets.}
  A plain net $N$ is \defitem{behaviourally $\quire$-asynchronous}
  iff there exists a distribution $D$ of $N$ meeting the requirement
  $\quire$ such that $I_D(N) \approx_{\mathscr{R}} N$.
\end{definitionG}

\noindent
 Intuitively, the only behavioural difference between a net $N$ and
 its asynchronous implementation $I_D(N)$ can occur when in $N$ a
 place $s \in \precond{u}$ is marked, whereas in $I_D(N)$ this token
 is already on its way from  $s$ to its posttransition $u$.
 In that case, it may occur that a transition $t\neq u$ with $s \in
 \precond{t}$ is enabled in $N$, whereas $t$ is not enabled in the described
 state of $I_D(N)$. We call the situation in $N$ leading to this
 state of $I_D(N)$ a \emph{distributed conflict}; it is in fact
 the only circumstance in which $I_D(N)$ fails to faithfully mimic the
 dynamic behaviour of $N$.

\begin{definitionG}{df-confusionfree}
  {Let $N = (S, T, F, M_0, \ell)$ be a net and $D$ a distribution of $N$.}
  $N$ has a \defitem{distributed conflict with respect to $D$} iff
  $$\exists t,u\inp T \;\exists p \inp \precond{t} \cap \precond{u}.
  t\neq u \wedge p \not\equiv_D u \wedge
  \exists M \inp [M_0\rangle_N. \precond{t} \subseteq M\trail{.}$$
\vspace{-2em}
\end{definitionG}

\noindent
 We wish to call a net $N$ \emph{(semi)structurally asynchronous} iff
 the situation outlined above never occurs, so that the asynchronous
 implementation does not change the behaviour of the net. As for
 behavioural asynchrony, this notion of asynchrony is
 parametrised by the set of allowed distributions.

\begin{definitionG}{df-structasync}
  {Let $\quire$ be a requirement on distributions of nets.}
  A net $N$ is \defitem{(semi)structurally $\quire$-asynchronous} iff there
  exists a distribution $D$ of $N$ meeting the requirement $\quire$
  such that $N$ has no distributed conflicts with respect to $D$.
\end{definitionG}

\noindent
The following theorem shows that distributed conflicts describe
exactly the critical situations: For all plain nets the notions of structural
and behavioural asynchrony coincide, regardless of the choice if $\quire$.

\begin{theoremG}{thm-plainbehstrcoincide}
  {Let $N$ be a plain net, and $\quire$ a requirement on
  distributions of nets.}
  Then $N$ is behaviourally $\quire$-asynchronous iff it is
  structurally $\quire$-asynchronous.
\end{theoremG}
\begin{proof}
  In Appendix~\ref{app-asynchronous}.
\end{proof}

\noindent
Because of this theorem, we call a plain net
$\quire$-asynchronous if it is behaviourally and/or structurally
$\quire$-asynchronous. In this paper we study this concept for
plain nets only. When taking $\quire=\quire_{FD}$ we speak
of \defitem{fully asynchronous nets}, when taking
$\quire=\quire_{SD}$ of \defitem{symmetrically asynchronous
  nets}, and when taking $\quire=\quire_{AD}$
of \defitem{asymmetrically asynchronous nets}.

\begin{example}
The net $N$ of \reffig{impl-result} is not fully asynchronous, for its
unique $D$-based asynchronous implementation $I_D(N)$ with $D\in
\quire_{FD}$ (also displayed in \reffig{impl-result}) is not step
readiness equivalent to $N$. In fact $\langle \varepsilon,
\emptyset\rangle \in \mathscr{R}(I_D(N))\setminus\mathscr{R}(N)$. This
inequivalence arises because in $I_D(N)$ the option to do an $a$-action
can be disabled already before any visible action takes place; this is
not possible in $N$.

The only way to avoid a distributed conflict in this net is by taking
$t \equiv_D p \equiv_D u$. This is not allowed for any $D \in
\quire_{FD}$ or $D \in \quire_{SD}$, but it is allowed for
$D \in \quire_{AD}$ (cf.\ the last net in \reffig{impl-result}).
Hence $N$ is asymmetrically asynchronous, but not symmetrically asynchronous.
\end{example}

\noindent
Since $\quire_{FD} \subseteq \quire_{SD} \subseteq \quire_{AD}$, any
fully asynchronous net is symmetrically asynchronous, and any
symmetrically asynchronous net is also asymmetrically asynchronous.
Below we give semi-structural characterisations of these three classes
of nets.
The first two stem from \cite{glabbeek08symmasymm}, where
the class of fully asynchronous nets is called $\FSA(B)$ and the class
of symmetrically asynchronous nets is called $\SA(B)$. The class
$\AA(B)$ in \cite{glabbeek08symmasymm} is somewhat larger than our
class of asymmetrically asynchronous nets, for it is based on a
slightly more involved protocol for replacing synchronous by
asynchronous interaction.

\pagebreak[2]
\begin{definitionG}{df-reachablem}
  {A plain net $N = (S, T, F, M_0, \ell)$ has a}
  \begin{itemise}
    \item \defitem{partially reachable conflict} iff
      $$\exists t,u\inp T \; \exists p \inp \precond{t} \cap \precond{u}.
      t\neq u \wedge 
      \exists M \inp [M_0\rangle_N. \precond{t} \subseteq M
      \trail{,}$$
    \item \defitem{partially reachable \structuralN} iff
      $$\exists t,u\inp T \; \exists p \inp \precond{t} \cap \precond{u}.
      t\neq u \wedge |\precond{u}|>1 \wedge 
      \exists M \inp [M_0\rangle_N. \precond{t} \subseteq M
      \trail{,}$$
    \item
      \defitem{left and right border reachable \structuralM} iff
      $$\exists t,u,v \inp T \;\exists p \inp \precond{t} \cap
      \precond{u}\; \exists q \inp \precond{u} \cap \precond{v}.
      \begin{array}{l} t \ne u \wedge u \ne v \wedge p \ne q \wedge {}\\
      \exists M_1, M_2 \inp
      [M_0\rangle_N. \precond{t} \subseteq M_1 \wedge \precond{v} \subseteq M_2
      \trail{.}
      \end{array}$$
  \end{itemise}
\vspace{-1em}
\end{definitionG}

\begin{theoremG}{thm-aabequalsweaklrmfree}
  {Let $N$ be a plain net.}
\begin{itemise}
\item $N$ is fully asynchronous iff it has no partially reachable conflict.
\item $N$ is symmetrically asynchronous iff it has no partially
  reachable $\structuralN$.
\item $N$ is asymmetrically asynchronous iff it has no left and
  right border reachable \structuralM.
\end{itemise}
\end{theoremG}
\begin{proof}
Straightforward with \refthm{plainbehstrcoincide}.
\end{proof}

\noindent
In the theory of Petri nets, there have been extensive studies on
classes of nets with certain structural properties like \textit{free
  choice nets} \cite{best83freesimple,bes87} and \textit{simple nets}
\cite{best83freesimple}, as well as extensions of theses classes. They
are closely related to the net classes defined here, but they are
defined without taking reachability into account. For a comprehensive
overview and discussion of the relations between those purely
structurally defined net classes and our net classes see
\cite{glabbeek08symmasymm}.
Restricted to plain nets without dead transitions (meaning that every
transition $t$ satisfies the requirement $\exists M \inp
[M_0\rangle. \precond{t} \subseteq M$), \refthm{aabequalsweaklrmfree}
says that a net is fully synchronous iff it is conflict-free in the structural sense (no shared preplaces),
symmetrically asynchronous iff it is a free choice net
and asymmetrically asynchronous iff it
is simple.  

Our asynchronous net classes are defined for plain nets only.
There are two approaches to lifting them to labelled nets.
One is to postulate that whether a net is asynchronous or not has
nothing to do with its labelling function, so that after replacing this
labelling by the identity function one can apply the insights above.
This way our structural characterisations
(Theorems~\ref{thm-plainbehstrcoincide}
and~\ref{thm-aabequalsweaklrmfree}) apply to labelled nets as well.
Another approach would be to apply the notion of behavioural asynchrony
of \refdf{behaviourasync} directly to labelled nets. This way more
nets will be asynchronous, because in some cases a net happens to be
equivalent to its asynchronous implementation in spite of a failure
of structural asynchrony. This happens for instance if all
transitions in the original net are labelled $\tau$.
Unlike the situation for plain nets, the resulting notion of
behavioural asynchrony will most likely be strongly dependent on the
choice of the semantic equivalence relation between nets.

\section{Distributed Systems}
\label{distributed systems}

The approach of Section~\ref{asynchronous} makes a difference between
a net regarded as a specification, and an asynchronous implementation
of the same net. The latter could be thought of as a way to execute
the net when a given distribution makes the synchronisations that are
inherent in the specification impossible.
In this and the following section, on the
other hand, we drop the difference between a net and its asynchronous
implementation.
Instead of adapting our intuition about the firing
rule when implementing a net in a distributed way, we insist that all
synchronisations specified in the original net remain present as
synchronisations in a distributed implementation. Yet, at the same
time we stick to the point of view that it is simply not possible for
a transition to synchronise its firing with the removal of tokens from
preplaces at remote locations. Thus we only allow distributions in
which each transition is co-located with all of its preplaces.  We
call such distributions \emph{effectual}.  For effectual distributions
$D$, the implementation transformation $I_D$ is the identity.  As a
consequence, if effectuality is part of a requirement $\quire$ imposed
on distributions, the question whether a net is $\quire$-asynchronous
is no longer dependent on whether an asynchronous implementation
mimics the behaviour of the given net, but rather on whether the net
allows a distribution satisfying $\quire$ at all.

The requirement of effectuality does not combine well will the
requirements on distributions proposed in \refdf{icerequirements}.
For if $\quire$ is the class of distributions that are effectual and
asymmetrically distributed, then only nets without transitions with
multiple preplaces would be $\quire$-asynchronous.  This rules out
most useful applications of Petri nets.  The requirement of
effectuality by itself, on the other hand, would make every net
asynchronous, because we could assign the same location to all places
and transitions.

We impose one more fundamental restriction on distributions, namely
that when two visible transitions can occur in one step, they cannot be
co-located. This is based on the assumption that at a given location
visible actions can only occur sequentially, whereas we want to
preserve as much concurrency as possible (in order not to loose
performance). Recall that in Petri nets simultaneity of transitions
cannot be enforced: if two transitions can fire in one step,
they can also fire in any order. The standard interpretation of nets
postulates that in such a case those transitions are causally
independent, and this idea fits well with the idea that they reside at
different locations.

\begin{definitionG}{df-}
  {Let $N = (S, T, F, M_0, \ell)$ be a net.}
  The \defitem{concurrency relation} $\mathord{\concurrent} \subseteq
  T^2$ is given by $t \concurrent u \equivalent t \ne u \wedge \exists
  M \inp [M_0\rangle. M [\{t,u\}\rangle$.

  $N$ is \defitem{distributed} iff it has a distribution $D$ such that
  \begin{itemise}
    \item
      $\mathrlap{\forall s \in S, ~t \in T.\hspace{1pt}s \in \precond{t}}
      \hphantom{t \concurrent u \wedge l(t),l(u) \ne \tau} \implies t \equiv_D s$,
    \item
     $t \concurrent u \wedge l(t),l(u) \ne \tau \implies t\not\equiv_D u$.
  \end{itemise}
\end{definitionG}

\noindent
It is straightforward to give a semi-structural characterisation of
this class of nets:

\begin{observation}{distributed}
A net is distributed iff there is no sequence $t_0,\ldots,t_n$ of
transitions with $t_0 \smile t_n$ and
$\precond{t_{i-1}}\cap\precond{t_{i}}\neq\emptyset$ for $i=1,\ldots,n$.
\end{observation}

\noindent
A structure as in the above characterisation of distributed
nets can be considered as a prolonged {\structuralM} containing two
independent transitions that can be simultaneously enabled.

It is not hard to find a plain net that is fully asynchronous, yet not
distributed. However, restricted to plain nets without dead
transitions, the class of asymmetrically asynchronous nets is a strict
subclass of the class of distributed nets. Namely, if a net is
\structuralM-free (where an {\structuralM} is as in
\refdf{reachablem}, but without the reachability condition on the
bottom line), then it surely has no sequence as described above.

\section{Distributable Systems}
\label{distributable}

\begin{figure}[tb]
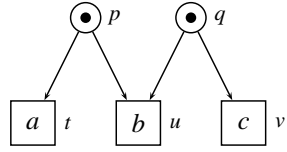

  \begin{center}
    \begin{petrinet}(6,3,4)
      \Q (2,3):p1:p;
      \Q (4,3):p2:q;
      \u (1,1):t1:a:t;
      \u (3,1):t2:b:u;
      \u (5,1):t3:c:v;

      \a p1->t1;
      \a p1->t2;
      \a p2->t2;
      \a p2->t3;
    \end{petrinet}
  \end{center}
\vspace{-1em}
  \caption{A fully marked \structuralM.}
  \label{fig-fullM}
\end{figure}

In this section, we will investigate the borderline for distributability
of systems. It is a well known fact that sometimes a global protocol
is necessary when concurrent activities in a system interfere. In
particular, this may be necessary for deciding choices in a coherent
way. Consider for example the simple net in \reffig{fullM}. It
contains an \structuralM-structure, which was already exhibited as a
problematic one in Section~\ref{asynchronous}. Transitions $t$ and $v$ are
supposed to be concurrently executable (if we do not want to restrict
performance of the system), and hence reside on different locations.
Thus at least one of them, say $t$, cannot be co-located with transition $u$.
However, both transitions are in conflict with $u$. 

As we use nets as models of reactive systems, we allow the environment
of a net to influence decisions at runtime by blocking one of the
possibilities. Equivalently we can say it is the environment that
fires transitions, and this can only happen for transitions that are
currently enabled in the net. If the net decides between $t$ and $u$
before the actual execution of the chosen transition, the environment
might change its mind in between, leading to a state of deadlock.
Therefore we work in a branching time semantics, in which the option
to perform $t$ stays open until either $t$ or $u$ occurs. Hence the
decision to fire $u$ can only be taken at the location of $u$, namely
by firing $u$, and similarly for $t$.  Assuming that it takes time to
propagate any message from one location to another, in no distributed
implementation of this net can $t$ and $u$ be simultaneously enabled,
because in that case we cannot exclude that both of them happen.
Thus, the only possible implementation of the choice between $t$ and
$u$ is to alternate the right to fire between
$t$ and $u$, by sending messages between them (cf.\ \reffig{fullMbusy}).
But if the
environment only sporadically tries to fire $t$ or $u$ it may
repeatedly miss the opportunity to do so, leading to an infinite
loop of control messages sent back and forth, without either
transition ever firing.

In this section we will formalise this reasoning, and show that under
a few mild assumptions this type of structures cannot be implemented in a
distributed manner at all, \ie even when we allow the implementation to be
completely unrelated to the specification, except for its behaviour.
For this, we apply the notion of a distributed net, as introduced in
the previous section. Furthermore, we need an equivalence notion in
order to specify in which way an implementation as a distributed net
is required to preserve the behaviour of the original net.
As in Section~3, we choose step readiness equivalence.
We call a plain net \emph{distributable}
if it is step readiness equivalent to a distributed net.
\linebreak[2]
We speak of a \textit{truly synchronous} net if it is not
distributable, thus if it may not be transformed into any
distributed net with the same behaviour up to step readiness equivalence, that
is if no such net exists. We study the concept ``distributable''
for plain nets only, but in order to get the largest class possible
we allow non-plain implementations, where a given transition may be
split into multiple transitions carrying the same label.

\begin{definition}{df-trulysync}
  A plain net $N$ is \defitem{truly synchronous} iff there exists no
  distributed net $N'$ which is step readiness equivalent to $N$.
\end{definition}

\noindent
We will show that nets like the one of \reffig{fullM} are truly
synchronous.

\begin{figure}[tb]
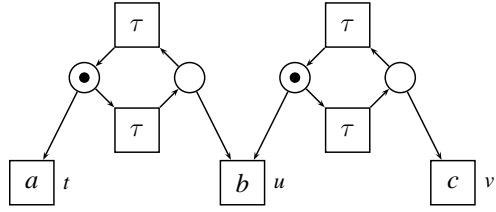

  \begin{center}
    \begin{petrinet}(10,4.5)
      \P (2,3):p;
      \p (4,3):pprime;
      \P (6,3):q;
      \p (8,3):qprime;

      \u (1,1):a:a:t;
      \u (5,1):b:b:u;
      \u (9,1):c:c:v;

      \t (3,2):ptau1:$\tau$;
      \t (3,4):ptau2:$\tau$;
      \t (7,2):qtau1:$\tau$;
      \t (7,4):qtau2:$\tau$;

      \a p->ptau1; \a ptau1->pprime; \a pprime->ptau2; \a ptau2->p;
      \a q->qtau1; \a qtau1->qprime; \a qprime->qtau2; \a qtau2->q;

      \a p->a;
      \a pprime->b;
      \a q->b;
      \a qprime->c;
    \end{petrinet}
  \end{center}
\vspace{-1em}
  \caption{A busy-wait implementation of the net in \reffig{fullM}}
  \label{fig-fullMbusy}
\end{figure}

Step readiness equivalence is one of the simplest and least
discriminating equivalences imaginable that preserves branching time,
causality and divergence to some small extend. Our impossibility
result, formalised below as \refthm{trulysyngltfullm},
depends crucially on all three properties, and thus needs to be
reconsidered when giving up on any of them.\linebreak[2]
When working in linear time semantics, every net is equivalent to an
infinite net that starts with a choice between several
$\tau$-transitions, each followed by a conflict-free net modelling a
single run. This net is $\structuralN$-free, and hence distributed.
It can be argued that infinite implementations are not acceptable, 
but when searching for the theoretical limits to distributed
implementability we don't want to rule them out dogmatically.
When working in interleaving semantics, any net can be converted into
an equivalent distributed net by removing all concurrency
between transitions. This can be accomplished by adding a new,
initially marked place, with an arc to and from every transition in
the net.
When fully abstracting from divergence, even when respecting causality
and branching time, the net of \reffig{fullM} is equivalent to the
distributed net of \reffig{fullMbusy}, and in fact it is not hard to
see that this type of implementation is possibly for any given net.
Yet, the implementation is suspect, as the implemented decision of a
choice may fail to terminate.
The clause $M \arrownot\production{\tau}$ in \refdf{readypair} is
strong enough to rule out this type of implementation, even though our
step readiness semantics abstracts from other forms of divergence.

We now characterise the class of nets which we will prove to be truly
synchronous.

\begin{definitionG}{df-fullM}
  {Let $N = (S, T, F, M_0, \ell)$ be a net.}
  $N$ has a \defitem{fully reachable visible pure \structuralM} iff
  $\exists t,u,v \in T.
  \precond{t} \cap \precond{u} \ne \varnothing \wedge
  \precond{u} \cap \precond{v} \ne \varnothing \wedge
  \precond{t} \cap \precond{v} = \varnothing \wedge\linebreak[2]
  \ell(t), \ell(u), \ell(v) \ne \tau \wedge \linebreak[1]
  \exists M \in [M_0\rangle.
  \precond{t} \cup \precond{u} \cup \precond{v} \subseteq M$.
\end{definitionG}

\noindent
Here a \emph{pure \structuralM} is an {\structuralM} as in
\refdf{reachablem} that moreover satisfies $\precond{t} \cap
\precond{v} = \varnothing$, and hence $p \not\in \precond{v}$,
$q \not\in \precond{t}$ and $t\neq v$.
These requirements follow from the conditions above.

\begin{proposition}{pr-fullmimpliesnotdistr}
  A net with a fully reachable visible pure {\structuralM} is not distributed.
\end{proposition}
\begin{proof}
  Let $N = (S, T, F, M_0, \ell)$ be a net that has a fully reachable
  visible pure \structuralM, so there exist 
  $t, u, v \in T$ and $p, q \in S$ such that
  $p \in \precond{t} \cap \precond{u} \wedge
  q \in \precond{u} \cap \precond{v} \wedge
  \precond{t} \cap \precond{v} = \varnothing$ and
  $\exists M \in [M_0\rangle.
  \precond{t} \cup \precond{u} \cup \precond{v} \subseteq M$.
  Then $t \concurrent v$.
  Suppose $N$ is distributed by the distribution $D$.
  Then $t \equiv_D p \equiv_D u \equiv_D q \equiv_D v$ but
  $t \concurrent v$ implies $t \not\equiv_D v$. \contradiction
\end{proof}

\noindent
Now we show that fully reachable visible pure {\structuralM}'s that
are present in a plain net are preserved under step readiness equivalence.

\begin{lemmaG}{lem-plainfullmimpliesreadym}
  {Let $N = (S, T, F, M_0, \ell)$ be a plain net.}
  If $N$ has a fully reachable visible pure \structuralM, there exists
  $\rpair{\sigma, X} \in \readyset(N)$ such that
  $\exists a, b, c \in \Act.\linebreak[2] a \neq c \wedge
  \{b\} \in X \wedge\linebreak[3] \{a, c\} \in X \wedge
  \{a, b\} \notin X \wedge \{b, c\} \notin X$.
  (It is implied that $a \ne b \ne c$.)
\end{lemmaG}
\begin{proof}
  $N$ has a {fully reachable visible pure \structuralM}, so
  there are $t,u,v \inp T$ and $M \inp [M_0\rangle$ such that
  $
  \precond{t} \cap \precond{u} \ne \varnothing\linebreak[2] \wedge
  \precond{u} \cap \precond{v} \ne \varnothing \wedge
  \precond{t} \cap \precond{v} = \varnothing \wedge
  \ell(t), \ell(u), \ell(v) \ne \tau \wedge
  \precond{t} \cup \precond{u} \cup \precond{v} \subseteq M$.
  Let $\sigma \in \Act^*$ such that $M_0 \Production{\sigma} M$.
  Since $N$ is a plain net, $M \arrownot\production{\tau}$ and
  $\ell(t)\neq\ell(u)\neq\ell(v)\neq\ell(t)$. Hence there exists an $X\! \subseteq \powermultiset{\Act}$
  such that $\rpair{\sigma, X}\! \in\! \readyset(N) \wedge \{\ell(u)\} \inp X
  \wedge\linebreak[2] \{\ell(t), \ell(v)\} \in X \wedge \{\ell(t), \ell(u)\} \notin X \wedge
  \{\ell(u), \ell(v)\} \notin X$.  \end{proof}

\begin{lemmaG}{lem-readymimpliesfullm}
  {Let $N = (S, T, F, M_0, \ell)$ be a net.}
  If there exists $\rpair{\sigma, X} \in \readyset(N)$ such that
  $\exists a, b, c \in \Act. a \neq c \wedge
  \{b\} \in X \wedge \{a, c\} \in X \wedge
  \{a, b\} \notin X\linebreak[2] \wedge \{b, c\} \notin X$,
  then $N$ has a fully reachable visible pure \structuralM.
\end{lemmaG}
\begin{proof}
  Let $M \subseteq S$ be the marking which gave rise to the step ready pair
\vspace{-2pt}
  $\rpair{\sigma, X}$, \ie $M_0 \Production{\sigma} M$ and
  $M \production{\{b\}} \wedge\, M \production{\{a,c\}} \wedge\,
  M \arrownot\production{\{a,b\}} \wedge\, M \arrownot\production{\{b,c\}}$.

  As $a \ne b \ne c \ne a$ there must exist three transitions $t, u, v \in T$
  with $\ell(t) = a\linebreak[2] \wedge \ell(u) = b \wedge \ell(v) = c$ and
  $M [\{u\}\rangle \wedge M [\{t,v\}\rangle \wedge
  \neg(M[\{t,u\}\rangle) \wedge \neg(M[\{u,v\}\rangle)$.
  From $M[\{u\}\rangle \wedge M[\{t,v\}\rangle$ follows
  $\precond{t} \cup \precond{u} \cup \precond{v} \subseteq M$.
  From $M[\{t,v\}\rangle$ follows $\precond{t} \cap \precond{v} = \varnothing$.
  From $\neg(M[\{t,u\}\rangle)$ then follows $\precond{t} \cap \precond{u} \ne
  \varnothing$ and analogously for $u$ and $v$.
  Hence $N$ has a fully reachable visible pure \structuralM.
\end{proof}

\noindent
Note that the lemmas above give a behavioural property that for plain
nets is equivalent to having a fully reachable visible pure \structuralM.

\begin{theorem}{thm-trulysyngltfullm}
  A plain net with a fully reachable visible pure {\structuralM} is
  truly synchronous.
\end{theorem}
\begin{proof}
  Let $N$ be a plain net which has a fully reachable visible pure \structuralM.
  Let $N'$ be a net which is step readiness equivalent to $N$.
  By \reflem{plainfullmimpliesreadym} and \reflem{readymimpliesfullm}, also
  $N'$ has a fully reachable visible pure \structuralM.
  By \refpr{fullmimpliesnotdistr}, $N'$ is not distributed. Thus $N$ is
  truly synchronous.
\end{proof}

\noindent
\refthm{trulysyngltfullm} gives an upper bound of the class of
distributable nets. We conjecture that this upper bound is tight, and
a plain net is distributable iff it has no fully reachable
  visible pure \structuralM.

\begin{conjecture}{cj-fullmlttrulysync}
  \hspace{-2.5pt}A plain net is truly synchronous iff it has a fully reachable
  visible pure \structuralM.$\!$
\end{conjecture}

\noindent
In the following, we give a lower bound of distributability by providing
a protocol to implement certain kinds of plain nets distributedly.
These implementations do not add additional labelled transitions, but
only provide the existing ones with a communication protocol in the
form of $\tau$-transitions. Hence these implementations pertain to a
notion of distributability in which we restrict implementations to be
plain $\tau$-nets. Note that this does not apply to the impossibility
result above.

\begin{definition}{df-distributable}
  A plain net $N$ is \defitem{plain-distributable} iff there exists
  a distributed plain $\tau$-net $N$ which is step readiness equivalent to $N$.
\end{definition}

\begin{definitionG}{df-conflconcur}
  {Let $N = (S, T, F, M_0, \ell)$ be a net.}
  We define the
  \defitem{enabled conflict relation} $\mathord{\conflict} \subseteq T^2$ as
  $$t \conflict u \equivalent \exists M \in [M_0\rangle. M [\{t\}\rangle \wedge
  M [\{u\}\rangle \wedge \neg (M [\{t, u\}\rangle).$$
\vspace{-2em}
\end{definitionG}

\begin{figure}[tb]
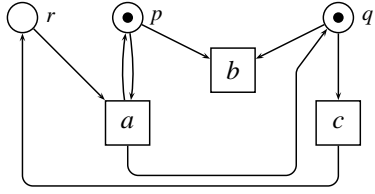

  \begin{center}
    \begin{petrinet}(8,4.4)
      \q (1,4):r:r;
      \Q (3,4):p:p;
      \Q (7,4):q:q;

      \t (3,2):a:a;
      \t (5,3):b:b;
      \t (7,2):c:c;

      \a r->a;
      \A p->a;
      \A a->p;
      \a p->b;
      \a q->b;
      \a q->c;
      \av a[-90]-(3,1)(6.2,1)(6.2,3)->[-135]q;
      \av c[-90]-(7,0.8)(1,0.8)->[-90]r;
    \end{petrinet}
  \end{center}
\vspace{-1.8em}
  \caption{An example net}
  \label{fig-distrorig}
\end{figure}

\begin{figure}[tb]
\begin{center}
\hspace{-1.1em}
  \begin{petrinet}(17,14.2)
    \q (2,12):r:r;
    \u (2,10):rsqr:$\tau$:$\boxed{r}$;
    \q (2,8):rla:$r^{[a]}$;
    \t (4,6):a:a;
    \q (4,4):acirc:$\circled{a}$;
    \u (4,2):aprime:$\tau$:$a'$;

    \Q (8,12):p:p;
    \u (8,10):psqr:$\tau$:$\boxed{p}$;
    \q (6,8):pla:$p^{[a]}$;
    \q (10,8):plb:$p^{[b]}$;

    \Q (14,12):q:q;
    \u (14,10):qsqr:$\tau$:$\boxed{q}$;
    \q (14,8):qlb:$q^{[b]}$;
    \t (12,6):b:b;
    \q (12,4):bcirc:$\circled{b}$;
    \u (12,2):bprime:$\tau$:$b'$;

    \t (16,6):c:c;
    \q (16,4):ccirc:$\circled{c}$;
    \u (16,2):cprime:$\tau$:$c'$;

    \u (6,4.5):bpla:$\tau$:$b_p^{[a]}$;
    \q (7,6):pbla:$\!p_b^{[a]}$;
    \q (6,3):palbprime:$\;\overline{p}_a^{[b]}$;
    \u (10,4.5):aplb:$\tau$:$a_p^{[b]}$;
    \q (9,6):palb:$\!p_a^{[b]}$;
    \q (10,3):pblaprime:$\;\overline{p}_b^{[a]}$;

    \a r->rsqr;
    \a rsqr->rla;
    \a rla->a;
    \a a->acirc;
    \a acirc->aprime;

    \a p->psqr;
    \a psqr->pla;
    \a psqr->plb;
    \a pla->a;
    \a plb->b;
    \a qlb->b;
    \a b->bcirc;
    \a bcirc->bprime;

    \a q->qsqr;
    \a qsqr->qlb;
    \a qlb->c;
    \a c->ccirc;
    \a ccirc->cprime;

    \a pla->bpla;
    \a pbla->bpla;
    \a bpla->pblaprime;
    \av b[-110]-(11.5,5)(9,7)->[45]pbla;
    \a pblaprime->bprime;

    \a plb->aplb;
    \a palb->aplb;
    \a aplb->palbprime;
    \av a[-70]-(4.5,5)(7,7)->[135]palb;
    \a palbprime->aprime;

    \av aprime[-80]-(4,1.3)(17,1.3)(17,12.8)(14,12.8)->[90]q;
    \av aprime[-110]-(3.8,1.1)(17.2,1.1)(17.2,13)(8,13)->[90]p;
    \av cprime[-90]-(16,0.7)(1,0.7)(1,13)(2,13)->[90]r;

    {
      \psset{linestyle=dashed}
      \psline(0.3,9)(17.7,9)
      \psline(4,9)(4,14)
      \psline(12,9)(12,14)
      \psline(8,9)(8,0)
    }
  \end{petrinet}
  \end{center}
  \caption{A distributed implementation for the net in \reffig{distrorig},
  partitioning into localities shown by dashed lines}
  \label{fig-distrimpl}
\end{figure}

\noindent
We now propose the following protocol for implementing nets.
An example depicting it can be found in \reffig{distrimpl}.
As locations we take the places in a given net,
 and the equivalence classes of
transitions that are related by the reflexive and transitive closure
of the enabled conflict relation. We locate every transition $t$ in
its equivalence class, whereas every place gets a private location.
Every place $s$ will have an embassy $s^{[t]}$ in every location $[t]$
where one of its posttransitions $t\in \postcond{s}$ resides.
As soon as $s$ receives a token, it will distribute this information
to its posttransitions by placing a token in each of these embassies.
The arc from $s$ to $t$ is now replaced by an arc from $s^{[t]}$ to
$t$, so if $t$ could fire in the original net it can also fire in the
implementation. So far the construction allows two transitions in
different locations that shared the precondition $s$ to fire
concurrently, although they were in conflict in the original net.
However, if this situation actually occurs, these transitions would
have been in an enabled conflict, and thus assigned to the same location.
The rest of the construction is a matter of garbage collection.
If a transition $t$ fires, for each of its preplaces $s$, all
tokens that are still present in the various embassies
of $s$ in locations $[u]$ need to be removed from there. This is done
by a special internal transition \plat{$t_s^{[u]}$}. Once all these
transitions (for the various choices of $s$ and $[u]$) have fired,
an internal transition $t'$ occurs, which puts tokens in all the
postplaces of $t$.

\begin{definitionG}{df-transctrlchoice}
  {Let $N = (S, T, F, M_0, \ell)$ be a net.}
  Let $[t] := \{u \in T \mid t \conflict^* u\}$.
  The transition-controlled-choice implementation of $N$ is defined to be
  the net $N' := (S \cup S^\tau, T \cup T^\tau, F', M_0, \ell')$ with
  \begin{align*}
    S^\tau := {}& \{s^{[t]} \mid s \in S, t \in \postcond{s}\} \cup
      \{\circled{t} \mid t \in T\} \cup {}\\
      &\{s_t^{[u]}, \overline{s}_t^{[u]} \mid
      s \inp S,~ t, u \inp \postcond{s}, [u] \ne [t]\} \\
    T^\tau := {}& \{\boxed{s} \mid s \in S\} \cup
      \{t' \mid t \in T\} \cup {}\\
      &\{t_s^{[u]} \mid s \inp S,~ t, u \inp \postcond{s}, [u] \ne [t]\} {}\\
    F' := {}& \{(s, \boxed{s}) \mid s \in S\} \cup {}\\
      &\{(\boxed{s}, s^{[t]}), (s^{[t]}, t) \mid
       s \in S, t \in \postcond{s}\} \cup {}\\
      &\{(t, \circled{t}), (\circled{t}, t') \mid t \in T\} \cup {}\\
      &\{(t', s) \mid t \in T, s \in \postcond{t}\} \cup {}\\
      &\{(t, s_t^{[u]}), (s_t^{[u]}\!, t_s^{[u]}), (t_s^{[u]}\!,
      \overline{s}_t^{[u]}), (\overline{s}_t^{[u]}\!, t'), (s^{[u]}\!,
      t_s^{[u]}) \mid s \inp S,~ t,u \inp \postcond{s},~ [u] \ne [t]\}
  \end{align*}
  $\ell' \restrictedto T = \ell$ and $\ell'(T^\tau) = \{\tau\}$.
\vspace{-1ex}
\end{definitionG}

\begin{theorem}{thm-distreqnoconcurconfl}
  A plain net $N$ is plain-distributable iff
  $\mathord{\conflict^*} \cap \mathord{\concurrent} = \varnothing$.
\end{theorem}
\begin{proof}
  ``$\Rightarrow$'':
  When implementing a plain net $N$ by a plain $\tau$-net $N'$ that is step readiness
  equivalent to $N$, the $\conflict$ and $\concurrent$ relations
  between the transitions of $N$ also exists between the corresponding
  visible transitions of $N'$. This is easiest to see when writing
  $a_N$, resp.\ $a_{N'}$, to denote a transition in $N$, resp.\ $N'$,
  with label $a$, which must be unique since $N$ is a plain net,
  resp.\ $N'$ a plain $\tau$-net.
  Namely if $a_N \!\conflict b_N$, then $N$ has a step ready pair
  $\rpair{\sigma,X}$ with $\{a\},\!\{b\}\inp X$ but
  $\{a,b\}\mathbin{\not\in} X$.
  This must also be a step ready pair of $N'$, and hence $a_{N'} \conflict
  b_{N'}$. Likewise, $a_N \concurrent b_N$ implies $a_{N'} \concurrent
  b_{N'}$.

Thus if $\mathord{\conflict^*} \cap \mathord{\concurrent} \ne \varnothing$
holds in $N$, then the same is the case for $N'$, and hence $N'$ is not
distributed by Observation~\ref{distributed}.

  ``$\Leftarrow$'':
  If $\mathord{\conflict}^* \cap \mathord{\concurrent} = \varnothing$,
  $N$ can be implemented as specified in \refdf{transctrlchoice}.
  In fact, the transition-controlled-choice implementation of any net $N$
  yields a net that is step readiness equivalent to $N$.
  See Appendix~\ref{transition-controlled-choice} for a formal proof
  of this claim. By construction, if $N$ is plain, its
  transition-controlled-choice implementation is a plain $\tau$-net.
  Moreover, if
  $\mathord{\conflict}^* \cap \mathord{\concurrent} = \varnothing$ it
  never happens that concurrent visible transitions are co-located,
  and hence the implementation will be distributed.
\pagebreak[3]
\end{proof}

\noindent
Our definition of distributed nets only enforces concurrent actions
to be on different locations if they are visible, and
our implementation in \refdf{transctrlchoice} produces nets which
actually contain concurrent unobservable activity at the same location.
If this is undesired it can easily be amended by adding a single marked
place to every location and connecting that place to every transition
on that location by a self-loop. While this approach will introduce new causality
relations, step readiness equivalence will not detect this.

\section{Conclusion}

In this paper, we have characterised different grades of
asynchrony in Petri nets in terms of structural and behavioural
properties of nets.
Moreover, we have given both an upper and a lower bound of distributability
of behaviours. In particular we have shown that some branching-time
behaviours cannot be exhibited by a distributed system.

We did not consider connections from transitions to their postplaces
as relevant to determine asynchrony and distributability. This is
because we only discussed contact-free nets where no synchronisation
by postplaces is necessary. In the spirit of \refdf{fsi} we could
insert $\tau$-transitions on any or all arcs from transitions to their
postplaces, and the resulting net would always be equivalent to the original.

We have already given a short overview on related work in the
introduction of this paper.  Most closely related to our approach are
several lines of work using Petri nets as a model of reactive systems.

As mentioned in Section~\ref{asynchronous}, classes of nets with
certain structural properties like \textit{free choice nets}
\cite{best83freesimple,bes87} and \textit{simple nets}
\cite{best83freesimple}, as well as extensions of theses classes, have
been extensively studied in Petri net theory, and are closely related
to the classes of nets defined here. In \cite{best83freesimple},
Eike Best and Mike Shields introduce various transformations between
free choice nets, simple nets and extended variants thereof. They use
``essential equivalence'' to compare the behaviour of different nets,
which they only give informally. This equivalence is insensitive to
divergence, which is relied upon in their transformations. It
also does not preserve concurrency, which makes it possible to
implement \emph{behavioural free choice nets}, that may feature a fully
reachable visible $\structuralM$, as free choice nets.  They continue to show
conditions under which liveness can be guaranteed for many of these classes.

In \cite{aalst98beyond}, Wil van der Aalst, Ekkart Kindler and J\"org Desel
introduce two extensions to extended simple nets, by excluding self-loops from
the requirements imposed on extended simple nets.
This however assumes a kind of ``atomicity'' of
self-loops, which we did not allow in this paper. In particular we do not
implicitly assume that a transition will not change the state of a place it
is connected to by a self-loop, since in case of deadlock, the temporary
removal of a token from such a place might not be temporary indeed.

In \cite{reisig82buffersync}, Wolfgang Reisig introduces a class of systems
which communicate using buffers and where the relative speeds of different
components are guaranteed to be irrelevant. The resulting nets are
simple nets. He then proceeds introducing a decision procedure for the problem
whether a marking exists which makes the complete system live.

Dirk Taubner has in \cite{taubner88zurverteiltenimpl} given various
protocols by which to implement arbitrary Petri nets in the OCCAM programming
language. Although this programming language offers synchronous communication
he makes no substantial use of that feature in the protocols, thereby
effectively providing an asynchronous implementation of Petri nets. He does not
indicate a specific equivalence relation, but is effectively using
linear-time equivalences to compare implementations to the specification.

The work most similar to our approach we have found is the one by Hopkins,
\cite{hopkins91distnets}. 
There he already classified nets
by whether they are implementable by a net distributed among different locations.
He uses an interleaving equivalence to compare an implementation to the
original net, and while allowing a range of implementations, he
does require them to inherit some of the structure of the original net.
  The net classes he describes in his paper are larger than those of
Section \ref{asynchronous} because he allows more general interaction
patterns, but they are incomparable with those of Section
\ref{distributable}. One direction of this inequality depends on his
choice of interleaving semantics, which allows the implementation in
\reffig{hopkins-added-concurrency}. The step readiness equivalence we use
does not tolerate the added concurrency and the depicted
net is not distributable in our sense.
The other direction of the inequality stems from the fact that we allow
implementations which do not share structure with the specification but only
emulate its behaviour. That way, the net in \reffig{distr-not-hopkins} can be
implemented in our approach as depicted.

\begin{figure}[tb]
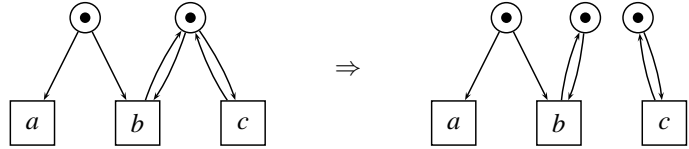

  \begin{center}
    \begin{petrinet}(14,3.4)
      \P (2,3):p;
      \P (4,3):q;

      \t (1,1):a:a;
      \t (3,1):b:b;
      \t (5,1):c:c;
      
      \a p->a; \a p->b;
      \A b->q; \A q->b; \A c->q; \A q->c;

      \rput(7,2){\Large $\Rightarrow$}

      \P (10,3):implp;
      \P (11.5,3):implqb;
      \P (12.5,3):implqc;
      
      \t (9,1):impla:a;
      \t (11,1):implb:b;
      \t (13,1):implc:c;

      \a implp->impla; \a implp->implb;
      \A implqb->implb; \A implb->implqb;
      \A implqc->implc; \A implc->implqc;
    \end{petrinet}
  \end{center}
\vspace{-1em}
  \caption{A specification and its Hopkins-implementation which added concurrency.}
  \label{fig-hopkins-added-concurrency}
\end{figure}

\begin{figure}
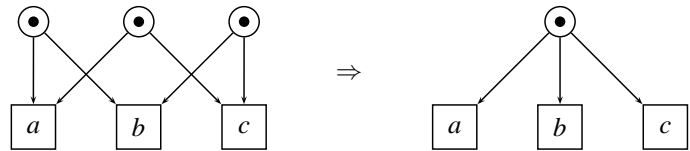

  \begin{center}
    \begin{petrinet}(14,3.4)
      \P (1,3):p;
      \P (3,3):q;
      \P (5,3):r;

      \t (1,1):a:a;
      \t (3,1):b:b;
      \t (5,1):c:c;

      \a p->a; \a p->b;
      \a q->a; \a q->c;
      \a r->b; \a r->c;

      \rput(7,2){\Large $\Rightarrow$}

      \P (11,3):pimpl;

      \t (9,1):aimpl:a;
      \t (11,1):bimpl:b;
      \t (13,1):cimpl:c;

      \a pimpl->aimpl; \a pimpl->bimpl; \a pimpl->cimpl;
    \end{petrinet}
  \end{center}
\vspace{-1em}
  \caption{A distributable net which is not considered distributable in
    \cite{hopkins91distnets}, and its implementation.}
  \label{fig-distr-not-hopkins}
\end{figure}

Still many open questions remain.
While our impossibility result holds even when allowing labelled nets as
implementations, our characterisation in \refthm{distreqnoconcurconfl}
only considers unlabelled
ones. This begs the question which class of nets can be implemented using
labelled nets. We conjecture that a distributed implementation exists for every
net which has no fully reachable visible pure \structuralM.
We also conjecture that if we allow
linear time correct implementations, all nets become distributable,
even when only allowing finite implementations of finite nets.
We are currently working on both problems.

Just as a distributable net is defined as a net that is
behaviourally equivalent to, or implementable by, a distributed net,
one could define an \emph{asynchronously implementable} net as one
that is implementable by an asynchronous net. This concept
is again parametrised by the choice of an interaction pattern.
It would be an interesting quest to characterise the various classes
of asynchronously implementable plain nets.

Also, extending our work to nets that are not required to be 1-safe
will probably generate interesting results, as conflict resolution
protocols must keep track of which token they are currently resolving
the conflict of.

In regard to practical applicability of our results, it would be very
interesting to relate our Petri net based terminology to hardware
descriptions in chip design. Especially in modern multi-core architectures 
performance reasons often prohibit using global clocks while a fa\c cade of
synchrony must still be upheld in the abstract view of the system.

On a higher level of applications, we expect our results to be useful
for language design. To start off, we would like to make a thorough
comparison of our results to those on communication patterns in
process algebras, versions of the $\pi$-calculus and I/O-automata
\cite{lynch96}. Using a Petri
net semantics of a suitable system description language, we could
compare our net classes to the class of nets expressible in the
language, especially when restricting the allowed communication
patterns in the various ways considered in~\cite{boer91embedding} or
in \cite{lynch96}.
Furthermore, we are interested in applying our results to graphical
formalisms for system design like UML sequence diagrams or activity
diagrams, also by applying their Petri net semantics. Our results
become relevant when such formalisms are used for the design of
distributed systems. Certain choice constructs become problematic
then, as they rely on a global mechanism for consistent choice
resolution; this could be made explicit in our framework.

\newpage
\appendix
\section{The Asynchronous Implementation}
\label{app-asynchronous}

Given a net $N$ and a distribution $D$ on $N$, this appendix explores
the properties of the $D$-based asynchronous implementation $I_D(N)$
of $N$, focussing on the relationship between $I_D(N)$ and $N$, and
culminating in the proofs of \refpr{} and \refthm{plainbehstrcoincide}
of Section~\ref{asynchronous}.

For better readability we will use the abbreviations $\iprecond{x} := \{y \mid (y,
x) \in F'\}$ and $\ipostcond{x} := \{y \mid (x, y) \in F'\}$ instead of
$\precond{x}$ or $\postcond{x}$ when making assertions about the flow relation
of an implementation.

The following lemma shows how the $D$-based asynchronous implementation of
a net $N$ simulates the behaviour of $N$.

\begin{lemmai}{Dbisim}{
  Let $N = (S, T, F, M_0, \ell)$ be a net, $A\subseteq
  \Act$, $\sigma\in \Act^*$ and $M_1, M_2 \subseteq S$.
  }
  \item\label{stepsim} If $M_1 \production{A}_{N} M_2$ then
  $M_1 \production{\tau}^*_{I_D(N)}\production{A}_{I_D(N)} M_2$.
  \item If $M_1 \Production{\sigma}_{N} M_2$ then $M_1 \Production{\sigma}_{I_D(N)} M_2$.
\end{lemmai}

\begin{proof}
  Assume $M_1 \,[G\rangle_N\, M_2$.  Then, by construction of $I_D(N)$,
  $$M_1 ~[\{t_s \mid t \inp G,~ s \inp \precond{t},~ s \not\equiv_D t\}\rangle_{I_D(N)}~
  [\{t \mid t \inp G\}\rangle_{I_D(N)}~ M_2.$$
  The first part of that execution can be split into a sequence of
  singleton transitions, all labelled $\tau$.
\\
  The second statement follows by a straightforward induction on the
  length of $\sigma$.
\end{proof}

This lemma uses the fact that any marking of $N$ is also a marking on
$I_D(N)$. The reverse does not hold, so in order to describe the
degree to which the behaviour of $I_D(N)$ is simulated by $N$ we need
to explicitly relate markings of $I_D(N)$ to those of $N$.  This is in
fact not so hard, as any reachable marking of $I_D(N)$ can be obtained
from a reachable marking of $N$ by moving some tokens into the newly
introduced buffering places $s_t$. To establish this formally, we
define a function which transforms implementation markings into the
related original markings, by shifting these tokens back.

\begin{definitionG}{tauback}{
  Let $N = (S, T, F, M_0, \ell)$ be a net
  and let $I_D(N) = (S \cup S^\tau, T\cup T^\tau, F', M_0, \ell')$.
  }
  $\tb: S \cup S^\tau \into S$ is the function defined by
  \begin{equation*}
    \tb(p) := \begin{cases}
      s & \text{ iff } p = s_t \text{ with } s_t \in S^\tau\!,~ s \in S,~ t \in T\\
      p & \text{ otherwise } (p \inp S)
    \end{cases}
  \end{equation*}
\vspace{-1em}
\end{definitionG}

Where necessary we extend functions to sets elementwise.
So for any $M \subseteq S \cup S^\tau$ we have $\tb(M)=\{\tb(s)\mid
s\in M\} = (M\cap S)\cup \{s \mid s_t\in M\}$. In particular,
$\tb(M)=M$ when $M\subseteq S$.

We now introduce a predicate $\alpha$ on the markings of $I_D(N)$
that holds for a marking iff it can be obtained from a reachable
marking of $N$ (which is also a marking of $I_D(N)$) by firing some
unobservable transitions. Each of these unobservable transitions moves
a token from a place $s$ into a buffering place $s_t$.\linebreak[2]
Later, we will show that $\alpha$ exactly characterises the reachable
markings of $I_D(N)$.
Furthermore, as every token can be moved only once, we can also give
an upper bound on how many such movements can still take place.
\pagebreak[2]

\begin{definitionG}{Dbivalidmarking}{
  Let $N = (S, T, F, M_0, \ell)$ be a net and $I_D(N) = (S \cup S^\tau,
  T \cup T^\tau, F', M_0, \ell')$.
  }
  The predicate $\fsivalidmarking \subseteq \powerset{S \cup S^\tau}$ is given by
  $$\fsivalidmarking(M) ~~:\equivalent~~
  \tb(M) \inp [M_0\rangle_N \wedge \forall p, q \inp M. \tb(p) = \tb(q) \implies p = q.$$
  The function $\fsidistance: \powerset{S \cup S^\tau} \into \IN \cup \{\infty\}$ is given by
  $\fsidistance(M) := |M \cap \{s \mid s \in S,~ \exists t \inp \postcond{s}.~
  s \not\equiv_D t\}|$, where
  we choose not to distinguish between different degrees of infinity.
\end{definitionG}

Note that $\fsivalidmarking(M)$ implies $|M|=|\tb(M)|$, and reachable
markings of $N$ are always finite (thanks to our definition of a net).
Hence $\fsivalidmarking(M)$ implies $d(M)\inp\IN$.  The following lemma
confirms that our informal description of $\fsivalidmarking$ matches
its formal definition.

\begin{lemmaG}{lem-Dbivalidmarking}{
  Let $N$ and $I_D(N)$ be as above and $M \subseteq S\cup S^\tau$,
  with $M$ finite.
  }
  Then $\forall p, q \inp M. \tb(p) = \tb(q) \implies p = q$ iff
  $\tb(M)\production{\tau}_{I_D(N)}^* M$.
\end{lemmaG}

\begin{proof} 
  Given that $\tb(M)\subseteq S$, ``if'' follows directly from the construction
  of $I_D(N)$.\\
  For ``only if'', assume $\forall p, q \inp M. \tb(p) = \tb(q) \implies p = q$.
  Then $\tb(M) ~[\{t_s \mid s_t \inp M\}\rangle_{I_D(N)}~ M$.
\end{proof}

Now we can describe how any net simulates the behaviour of
its fully asynchronous implementation.

\begin{lemmai}{Dbiworking}{
  Let $N$ and $I_D(N)$ be as above, $A\subseteq \Act$, $\sigma\in \Act^*$
  and $M,M' \subseteq S \cup S^\tau$.
  }
    \item\label{Dbiinvarstart} $\fsivalidmarking(M_0)$.
    \item\label{Dbiimplstep} If $\fsivalidmarking(M) \wedge M
      \production{A}_{I_D(N)} M'$ then
      $\tb(M) \production{A}_N \tb(M') \wedge \fsivalidmarking(M')$.
    \item\label{Dbiimpltaustep} If $\fsivalidmarking(M) \wedge M \production{\tau}_{I_D(N)} M'$
    then $\fsidistance(M) > \fsidistance(M') \wedge \tb(M) = \tb(M') \wedge \fsivalidmarking(M')$.
    \item\label{Dbiimpllongstep} If $M_0 \Production{\sigma}_{I_D(N)}
      M'$ then $M_0 \Production{\sigma}_N \tb(M')\wedge \fsivalidmarking(M')$.
\end{lemmai}
\begin{proof}
  \refitem{Dbiinvarstart}: $M_0 \in [M_0\rangle_N$ and $\forall s \in M_0 \subseteq S. \tb(s) = s$.

  \refitem{Dbiimplstep}: Suppose $\fsivalidmarking(M)$ and $M
    ~[G\rangle_{I_D(N)}~ M'$ with $G\subseteq T$.
    So $\tb(M)$ is a reachable marking of $N$.

    Note that for any $t\inp T$ we have that $\tb(\iprecond{t}) = \precond{t}$.  
    Moreover, $\alpha(M)$ implies that
\vspace{-1ex}
\begin{equation}\label{disjoint}
    X,Y \subseteq M \wedge X\cap Y = \emptyset ~~\Rightarrow~~ \tb(X)\cap\tb(Y)=\emptyset
\vspace{-1em}
\end{equation}
and hence
\vspace{-1ex}
\begin{equation}
    Y \subseteq M ~~\Rightarrow~~\tb(M\setminus
    Y)=\tb(M)\setminus\tb(Y) \trail{.}
\end{equation}
    Let $t \in G$. Since $t$ is enabled in $M$, we have $\iprecond{t}
    \subseteq M$ and hence $\precond{t} = \tb(\iprecond{t}) \subseteq
    \tb(M)$.  Given that $N$ is contact-free and $\tb(M) \in [M_0\rangle_N$,
    it follows that $t$ is enabled in $\tb(M)$.

    Now let $t, u \in G$ with $t \ne u$.  Then $\iprecond{t} \cup
    \iprecond{u} \subseteq M$ and $\iprecond{t} \cap \iprecond{u} =
    \varnothing$, so $\precond{t} \cap \precond{u} = \tb(\iprecond{t})
    \cap \tb(\iprecond{u}) = \emptyset$, using (\ref{disjoint}).
    Given that $\precond{t} \cup \precond{u} \subseteq \tb(M)$ and $N$
    is contact-free, it follows that also $\postcond{t} \cap
    \postcond{u} = \varnothing$ and hence $t$ and $u$ are independent.

    Since $M'= (M \setminus \iprecond{G}) \cup \ipostcond{G}$ we have
      $\tb(M')
      = (\tb(M) \setminus \tb(\iprecond{G})) \cup \tb(\ipostcond{G})
      = (\tb(M) \setminus \precond{G}) \cup \postcond{G}$ and hence
    $\tb(M) ~[G\rangle_N~ \tb(M')$.

    Next we establish $\fsivalidmarking(M')$. To this end, we may
    assume that $G$ is a singleton set, for $G$ must be finite---this
    follows since all (independent) transitions in $G$ are enabled
    from the reachable marking $\tb(M)$ of $N$, and $N$ satisfies the
    finiteness restrictions imposed on nets in Section~\ref{basic}---and when
    $M[\{t_0,t_1,\ldots,t_n\}\rangle M'$ for some $n\geq 0$ then there
    are $M_1,M_2,\ldots,M_n$ with $M\,[\{t_0\}\rangle\, M_1[\{t_1\}\rangle\, M_2
    \cdots M_n\,[\{t_n\}\rangle\, M'$, allowing us to obtain the general
    case by induction.  So let $G=\{t\}$ with $t\inp T$.

    Above we have shown that $\tb(M')\inp[M_0\rangle_N$.  We still
    need to prove that $\tb(p) = \tb(q) \implies p = q$ for all $p, q
    \inp M'$.  Assume the contrary, i.e.\ there are $p, q \in M'$ with
    $\tb(p) = \tb(q)$ but $p \ne q$.  Since $\fsivalidmarking(M)$, at
    least one of $p$ and $q$---say $p$---must not be present in $M$.
    Thus $p \in \ipostcond{t}=\postcond{t}\subseteq S$.  As
    $\tb(q)=\tb(p)=p$ and $q\neq p$, it must be that $q\in S^\tau$.
    Hence $q\mathbin{\notin}\ipostcond{t}$, so $q\inp M$, and
    $p=\tb(q)\in\tb(M)$.  As shown above, $t$ is enabled in $\tb(M)$.
    By the contact-freeness of $N$, $(\tb(M)\setminus \precond{t})\cap
    \postcond{t}=\emptyset$, so $p \inp \precond{t}$.  Using that
    $p\mathbin{\not\in}M$, we find that $p \mathbin{\not\in}
    \iprecond{t} \subseteq M$, so $p \not\equiv_D t$ and $p_t \in
    \iprecond{t} \subseteq M$. As by construction
    $\iprecond{t}\cap\ipostcond{t}=\emptyset$, we have $p_t\not\in
    M'$, so $q\neq p_t$. Yet $\tb(q)=p=\tb(p_t)$, contradicting
    $\fsivalidmarking(M)$.

  \refitem{Dbiimpltaustep}: Let $t_s \in T^\tau$ such that $M \,[\{t_s\}\rangle_{I_D(N)}\, M'$.
    Then, by construction of $I_D(N)$, $\iprecond{t_s} = \{s\} \wedge \ipostcond{t_s} = \{s_t\}$.
    Hence $M' = M \setminus \{s\} \cup \{s_t\}$ and
    $\fsidistance(M') = \fsidistance(M) - 1 \wedge \tb(M') = \tb(M)$.
    Moreover, $\fsivalidmarking(M') \equivalent \fsivalidmarking(M)$.

  \refitem{Dbiimpllongstep}: Using \ref{Dbiinvarstart}--\ref{Dbiimpltaustep},
   this follows by a straightforward induction on the number of
   transitions in the derivation $M_0 \Production{\sigma}_{I_D(N)} M'$.
\end{proof}

It follows that $\fsivalidmarking$ exactly characterises the
reachable markings of $I_D(N)$:

\begin{lemmaG}{Dbionesafe}{
  Let $N$ and $I_D(N)$ be as before and $M \subseteq S \cup S^\tau$.
  }
    Then $M \in [M_0\rangle_{I_D(N)}$ iff $\fsivalidmarking(M)$.
\end{lemmaG}

\begin{proof}
  ``Only if'' follows from \reflem{Dbiworking}.\ref{Dbiimpllongstep},
   and ``if'' follows by Lemmas~\ref{lem-Dbisim} and~\ref{lem-Dbivalidmarking}.
\end{proof}

Using this we now prove \refpr{} from Section~\ref{asynchronous}:
  
\setcounter{prop}{0}
\begin{proposition}{pr-2}
  For any (contact-free) net $N=(S,T,F,M_0,\ell)$, and any choice of $\equiv_D$, the
  net $I_D(N)$ is contact-free, and satisfies the other requirements
  imposed on nets, listed in Section~\ref{basic}.
\end{proposition}

\begin{proof}
  Let $M\in[M_0\rangle_{I_D(N)}$. Then
  $\fsivalidmarking(M)$, and hence $\tb(M) \in [M_0\rangle_N$.

 Consider any $t \in T$ with $\iprecond{t}\subseteq M$. Assume $(M
    \setminus \iprecond{t}) \cap \ipostcond{t} \ne \varnothing$. Since
    $\ipostcond{t} = \postcond{t} \subseteq S$ let $p \in S$ be such
    that $p \in M \cap \ipostcond{t}$ and $p \not\in \iprecond{t}$.
    As $N$ is contact-free we have $(\tb(M) \setminus \precond{t})
    \cap \postcond{t}=\emptyset$, so since $p \in \tb(M) \cap
    \postcond{t}$ it must be that $p \inp \precond{t}$.
    Hence $p_t \inp \iprecond{t} \subseteq M$ and we have $p \neq p_t$
    yet $\tb(p) \mathbin= p \mathbin= \tb(p_t)$, violating $\fsivalidmarking(M)$.

    Now consider any $t_p \in T^\tau$ with $\iprecond{t_p}\subseteq M$.
    As $\iprecond{t_p} = \{p\}$ and $\ipostcond{t_p} = \{p_t\}$
    we have that $(M \setminus \iprecond{t_p}) \cap \ipostcond{t_p} \ne \varnothing$
    only if $p \in M \wedge p_t \in M$. However, $\tb(p) = p = \tb(p_t)$ which would
    violate $\fsivalidmarking(M)$.

  This established the contact-freeness of $I_D(N)$. By construction,
  $M_0$ is finite, $\iprecond{t}\neq\emptyset$ and $\iprecond{t}$ and
  $\ipostcond{t}$ are finite for all $t\inp T \cup T^\tau$, and
  $\ipostcond{s}$ is finite for all $s\inp S \cup S^\tau$.
\end{proof}

The following lemma is a crucial step in the proof of
\refthm{plainbehstrcoincide}.

\begin{lemmaG}{lem-distributed conflict}
  {Let $N=(S,T,F,M_0,\ell)$ be a net without a distributed conflict
  w.r.t.\ a distribution $D$.}  Let $M_1 \inp [M_0\rangle_N$
  and $M_1 \production{\tau}_{I_D(N)} M_2 \production{\tau}_{I_D(N)}
  \cdots \production{\tau}_{I_D(N)} M_n
  \arrownot\production{\tau}_{I_D(N)}$ for some $n\geq 1$.
  Then, \plat{$M_1 \production{A}_N$ iff $M_n
  \production{A}_{I_D(N)}$} for all $A\subseteq \Act$.
\end{lemmaG}

\begin{proof}
  Suppose $\precond{t}\subseteq M_1$ but $\iprecond{t}\not\subseteq
  M_n$ for some $t\inp T$.  For $p\inp \precond{t}$ write $\hat{p}_t
  := p_t$ if $p \not\equiv_D t$ and $\hat{p}_t=p$ otherwise. Then
  $\iprecond{t} = \{\hat{p}_t \mid p \in \precond{t}\}$. Pick $p \inp
  \precond{t}$ such that $\hat{p}_t \mathbin{\not\in} M_n$. As
  $M_n\arrownot\production{\tau}_{I_D(N)}$ we also have $p \mathbin{\not\in} M_n$.
  Let $1\leq i<n$ be the last index such that $p \inp M_i$ or $\hat{p}_t \inp
  M_i$. Then $M_i ~[\{u_p\}\rangle_{I_D(N)}~ M_{i+1}$ for some $u \inp T$ with
  $u \ne t$, $p \inp \precond{u}$ and $p \not\equiv_D u$. But this would
  constitute a distributed conflict w.r.t.\ $D$. $\lightning$

  It follows that $M_1 \,[t\rangle_N$ implies $M_n \,[t\rangle_{I_D(N)}$
  for all $t \inp T$. Moreover, it follows immediately from the
  construction of $I_D(N)$ that if two transitions $t,u\inp T$ are
  independent in $N$, then they are also independent in $I_D(N)$.
  Hence $M_1 \,[G\rangle_N$ implies $M_n \,[G\rangle_{I_D(N)}$ for all
  $G \subseteq T$. Thus \plat{$M_1 \production{A}_N$ implies $M_n
  \production{A}_{I_D(N)}$}.

  For the reverse direction, observe that $\fsivalidmarking(M_1)$ and
  $\tb(M_1)=M_1$ because $M_1 \inp [M_0\rangle_N$. Hence
  $\fsivalidmarking(M_n)$ and $\tb(M_n)=M_1$ by
  \reflem{Dbiworking}.\ref{Dbiimpltaustep}
  and \plat{$M_n \production{A}_{I_D(N)}$ implies $M_1
  \production{A}_N$} for all $A$ by
  \reflem{Dbiworking}.\ref{Dbiimplstep}.
\end{proof}

\setcounter{theo}{0}
\begin{theoremG}{}
  {Let $N = (S,T,F,M_0,\ell)$ be a plain net, and $\quire$ a
  requirement on distributions of nets.}
  Then $N$ is behaviourally $\quire$-asynchronous iff it is
  structurally $\quire$-asynchronous.
\end{theoremG}
\begin{proof}
``Only if'': Suppose $N$ fails to be structurally
  $\quire$-asynchronous. Let $D$ be a distribution on $N$ meeting the
  requirement $\quire$.
  Then $N$ has a distributed conflict with respect to $D$,
  i.e.\ $$\exists t,u\inp T \;\exists p \inp \precond{t} \cap \precond{u}.
  t\neq u \wedge p \not\equiv_D u \wedge
  \exists M \inp [M_0\rangle_N. \precond{t} \subseteq M\trail{.}$$
  We need to show that $I_D(N) \not \approx_{\mathscr{R}} N$.

  Let $M\inp [M_0\rangle_N$ be such that $\precond{t} \subseteq M$ and
  let $\sigma\in\Act^*$ be such that $M_0 \Production{\sigma}_N M$.
  Then $N$ has a step ready pair $\rpair{\sigma,X}$ with $\{\ell(t)\} \inp
  X$. As plain nets are deterministic, $M$ is the only marking of $N$
  with the property that $M_0\Production{\sigma}_N M$.  Hence
  $N$ has exactly once step ready pair of the form $\rpair{\sigma,X}$,
  and it satisfies $\{\ell(t)\} \inp X$.

  \reflem{Dbisim} yields \plat{$M_0 \Production{\sigma}_{I_D(N)} M$}.
  Let $M_1 := (M \setminus \{p\}) \cup \{p_u\}$. Then \plat{$M
  ~[u_p\rangle_{I_D(N)}~ M_1$} by \refdf{fsi}, so $M \production{\tau}
  M_1$. By \reflem{Dbiworking}.\ref{Dbiimpltaustep}, we have $M_1
  \production{\tau}_{I_D(N)} M_2 \production{\tau}_{I_D(N)} \cdots
  \production{\tau}_{I_D(N)} M_n \arrownot\production{\tau}_{I_D(N)}$
  for some $n\leq d(M)\inp\IN$.  As $\ipostcond{v} \subseteq S^\tau$
  for all $v \inp T^\tau$, we have $p \mathbin{\not\in} M_i$ for
  $i=1,2,\ldots,n$. Moreover, in case $p \not\equiv t $ we have $p_t
  \inp \ipostcond{v}$ only if $p \inp \iprecond{v}$; hence also $p_t
  \mathbin{\not\in} M_i$ for $i=1,2,\ldots,n$. It follows that
  $\iprecond{t} \not\subseteq M_n$. Thus $I_D(N)$ has a step ready
  pair $\rpair{\sigma,X}$ with $\{\ell(t)\}\not\in X$.
  We find that $\readyset(I_D(N))\neq\readyset(N)$.

  ``If'': Suppose $N$ is structurally $\quire$-asynchronous, i.e.\
  there is a distribution $D$ on $N$ meeting the requirement $\quire$,
  such that $N$ has no distributed conflicts with respect to $D$.
  We show that $\readyset(I_D(N)) = \readyset(N)$.

  `'$\supseteq$'': Let $\rpair{\sigma,X} \inp \readyset(N)$.  Then
  there is a marking $M$ of $N$ such that $M_0 \Production{\sigma}_N
  M$, \plat{$M \production{A}_N$} for all $A \inp X$ and \plat{$M
  \arrownot\production{A}_N$} for all $A \mathbin{\not\in} X$.
  \reflem{Dbisim} yields \plat{$M_0 \Production{\sigma}_{I_D(N)} M$}.
  By \reflem{Dbiworking}.\ref{Dbiimpltaustep}, we have
  $M\production{\tau}_{I_D(N)} M_1 \production{\tau}_{I_D(N)} M_2
  \production{\tau}_{I_D(N)} \cdots \production{\tau}_{I_D(N)} M_n
  \arrownot\production{\tau}_{I_D(N)}$ for some $0\leq n\leq
  d(M)\inp\IN$.  Now \reflem{distributed conflict} yields
  $\rpair{\sigma,X} \inp \readyset(I_D(N))$.

  ``$\subseteq$'': Let $\rpair{\sigma,X} \inp \readyset(I_D(N))$.
  Then there is a marking $M$ of $I_D(N)$ such that $M_0
  \Production{\sigma}_{I_D(N)} M$, \plat{$M\arrownot\production{\tau}_{I_D(N)}$},
  and \plat{$M \production{A}_{I_D(N)}$} iff $A \inp X$.
  \reflem{Dbiworking}.\ref{Dbiimpllongstep} yields
  \plat{$M_0 \Production{\sigma}_{N} \tb(M) \wedge \fsivalidmarking(M)$} and
  \reflem{Dbivalidmarking} gives $\tb(M)\production{\tau}_{I_D(N)}^* M$.
  Now \reflem{distributed conflict} yields $\rpair{\sigma,X} \inp \readyset(N)$.
\end{proof}

\section{The Transition-Controlled-Choice Implementation}
\label{transition-controlled-choice}

In this appendix we show that the transition-controlled-choice
implementation of any net $N$ is step readiness equivalent to $N$.
To this end we use the following result.

\setcounter{equation}{0}
\begin{lemmaG}{lem-branching-bisimulation}{
  Let $N = (S, T, F, M_0, \ell)$ and $N' = (S'\!, T'\!, F'\!, M'_0, \ell')$
  be two nets, and $\ell'(t)\mathbin{\ne}\tau$ for $t \inp T'$.}
  Suppose there is a function $\tbtci\!: \powerset{S} \into
  \powerset{S'}$ from the markings of $N$ to the markings of $N'$,\linebreak[2]
  a \emph{distance} function $\tcidistance: \powerset{S} \into \IN \cup \{\infty\}$
  and a predicate $\tcivalidmarking \subseteq \powerset{S}$ such that
  \begin{align}
    & \tcivalidmarking(M_0) \wedge \tbtci(M_0) = M'_0\\
    & \plat{$\tcivalidmarking(M_1) \wedge M_1 \production{\tau}_{N} M_2 ~~\implies~~
      \tcivalidmarking(M_2) \wedge \tbtci(M_2) = \tbtci(M_1) \wedge
      \tcidistance(M_1) > \tcidistance(M_2)$}\\
    & \plat{$\tcivalidmarking(M_1) \wedge M_1 \production{A}_{N} M_2 ~~\implies~~
      \tcivalidmarking(M_2) \wedge \tbtci(M_1) \production{A}_{N'} \tbtci(M_2)$}\\
    & \plat{$\tcivalidmarking(M_1) \wedge \tcidistance(M_1) > 0 ~~\implies~~
      M_1 \production{\tau}_{N}$}\\
    & \plat{$\tcivalidmarking(M_1) \wedge \tcidistance(M_1) = 0 \wedge
      \tbtci(M_1) \production{A}_{N'} M'_2 ~~\implies~~ \exists M_2.\,
      M_1 \production{A}_{N} M_2 \wedge M'_2 = \tbtci(M_2)$} \trail{.}
  \end{align}
  Then $N \approx_\mathscr{R} N'$.
\end{lemmaG}

\begin{proof}
  ``$\readyset(N) \subseteq \readyset(N')$'': Conditions (1--5) allow
  any step ready pair $\rpair{\sigma,X}$ of $N$ to be mimicked step
  for step by $N'$. To be precise, if $\rpair{\sigma,X}\inp \readyset(N)$,
  then there is a marking $M_1$ with $M_0 \Production{\sigma}_N M_1$,
  \plat{$M_1 \arrownot\production{\tau}_N$}, \plat{$M_1 \production{A}_N$} for any
  $A \inp X$ and \plat{$M_1 \arrownot\production{A}_N$} for any $A \mathbin{\not\in} X$.
  As for all reachable markings $M_1$ of $N$, we have $\tcivalidmarking(M_1)$.
  Now (1--3) imply $M'_0 \Production{\sigma}_{N'} \tbtci(M_1)$.
  Furthermore, (3) implies \plat{$\tbtci(M_1) \production{A}_{N'}$} for
  any $A \inp X$, (4) implies $\tcidistance(M_1)=0$, and hence (5) implies
  \plat{$\tbtci(M_1) \arrownot\production{A}_{N'}$} for any $A \mathbin{\not\in} X$.

  ``$\readyset(N') \subseteq \readyset(N)$'': From conditions (2--5)
  we infer:
  \begin{align}
    & \plat{$\tcivalidmarking(M_1) ~~\implies~~ \exists M_2.\,
      M_1 \production{\tau}_N^* M_2 \wedge M_2
      \arrownot\production{\tau} \wedge~ \tcivalidmarking(M_2) \wedge
      \tbtci(M_2)=\tbtci(M_1)$}\\
    & \plat{$\tcivalidmarking(M_1) \wedge
      \tbtci(M_1) \production{A}_{N'} M'_2 ~~\implies~~ \exists M_2.\,
      M_1 \production{\tau}_N^*\production{A}_{N} M_2 \wedge
      \tcivalidmarking(M_2) \wedge M'_2 = \tbtci(M_2)$}
  \end{align}
  The first statement follows by repeated application of (2);
  the second by repeated application of (4) and (2), then (5) and (3).
  Conditions (1) and (7) imply that every reachable marking of $N'$ is
  of the form $\tbtci(M)$ with $M$ a reachable marking of $N$.
  Moreover, (1), (6) and (7) yield, for $\sigma \inp \Act^*$,
  \begin{align*}
    & \plat{$M'_0 \Production{\sigma}_{N'} M' ~~\implies~~ \exists M.\,
      M_0 \Production{\sigma}_{N} M \wedge M
      \arrownot\production{\tau} \wedge~
      \tcivalidmarking(M) \wedge M' = \tbtci(M)$}\trail{.}
  \end{align*}
  In combination with (3--5) this implies that any ready pair
  $\rpair{\sigma,X}$ of $N'$ is also a ready pair of $N$. 
\end{proof}

In fact, conditions (1--5) are strong enough to show that $N$ and $N'$
are semantically equivalence in various other ways as well; in
particular $\tbtci$ constitutes a \emph{branching bisimulation}
between $N$ and $N'$, as defined in \cite{GW96}.
In order to apply \reflem{branching-bisimulation}, we will take $N$ to
be the transition-controlled-choice implementation of a given net $N'$
that features no transitions labelled $\tau$. 

\begin{definitionG}{df-tciimplobjects}{}
  Let $N' = (S, T, F', M_0, \ell')$ be a net with
  $\ell'(t)\mathbin{\ne}\tau$ for $t \inp T'$, and $N = (S \cup S^\tau, T
  \cup T^\tau, F, M_0, \ell)$ its transition-based-choice implementation.

  The function $\tbtci: \powerset{S \cup S^\tau} \into \powerset{S}$ is defined by
  \begin{align*}
    \tbtci(M) := (M \cap S) \cup \{s \mid s \in S,~
      \{s^{[t]} \mid t \in \postcond{s}\} \subseteq M\} \cup
    \{s \mid s \in \postcond{t} \wedge \circled{t} \in M\} \trail{.}
  \end{align*}
  The function $\tbsafetci: \powerset{S \cup S^\tau} \into \powerset{S}$ is defined by
  \begin{align*}
    \tbsafetci(M) := (M \cap S) \cup \{s \mid s \in S,
      \{s^{[t]} \mid t \in \postcond{s}\} \cap M \ne \varnothing\} \cup
    \{s \mid s \in \precond{t} \wedge \circled{t} \in M\} \trail{.}
  \end{align*}
  The function $\tcidistance: \powerset{S \cup S^\tau} \into \IN \cup \{\infty\}$ is
  defined by
  \[
    \tcidistance(M) :=
    |M \cap S| + \sum_{\psscalebox{0.5}{\circled{t}} \in M} (1 + |\postcond{t}|) +
      \sum_{s_t^{[u]} \in M} 1 \trail{.}
  \vspace{-1em}
  \]
  The predicate $\tcivalidmarking \subseteq \powerset{S\cup S^\tau}$ is defined by
  \setcounter{equation}{0}
  \def\theequation{$\beta_\arabic{equation}$}
  \begin{align}
    \!\!\!\tcivalidmarking(M) :\equivalent {}
\label{B1} &\tbsafetci(M) \in [M_0\rangle_{N'} \wedge {} \\
\label{B2} &\left(s^{[t]} \in M \implies s \notin M\right) \wedge {} \\
\label{B3} &\left(s^{[u]} \inp M \wedge s^{[t]} \notinp M \implies
      \exists v \in \postcond{s}. s_v^{[u]} \in M\right) \wedge {}\\
\label{B4} & \left(\circled{t},\circled{u} \in M \wedge t\ne u \implies
      \precond{t} \cap \precond{u} = \emptyset\right) \wedge {}\\
\label{B5} &\left(s_t^{[u]} \in M \implies
      \overline{s}_t^{[u]} \notin M \wedge s^{[u]}, \circled{t} \in M\right) \wedge {}\\
\label{B6} &\left(\overline{s}_t^{[u]} \in M \implies \circled{t} \in M\right) \wedge {} \\
\label{B7} &\left(\circled{t} \in M \implies
      \forall s \inp \precond{t},\, u \inp \postcond{s}. s, s^{[t]} \notinp M \wedge
      \left([u] \ne [t] \implies
          s_t^{[u]} \inp M \vee \overline{s}_t^{[u]} \inp M\right)\right).
  \end{align}
\end{definitionG}
Some conjuncts in the definition of $\tcivalidmarking(M)$
  are universally quantified over (some of) $s$, $t$ and $u$;
  we write
  \begin{itemise}
  \item ${\beta_i}\!^{s,t,u}(M)$ to say that marking $M$
  satisfies the instance of $\beta_i$ for the specific values
  $s$, $t$ and $u$,
  \item ${\beta_i}\!^{s}(M)$ for $\forall t,u\inp
  \postcond{s}$. ${\beta_i}\!^{s,t,u}(M)$,
  \item and ${\beta_i}(M)$ for $\forall s\inp S$. ${\beta_i^s}(M)$,
  \vspace{1ex}
  \end{itemise}
  so that $\tcivalidmarking(M)$ iff $\mbox{\ref{B1}}(M) \wedge \mbox{\ref{B2}}(M)
  \wedge \mbox{\ref{B3}}(M) \wedge \mbox{\ref{B4}}(M) \wedge
  \mbox{\ref{B5}}(M) \wedge \mbox{\ref{B6}}(M) \wedge \mbox{\ref{B7}}(M)$.

\begin{lemmaG}{lem-tciimplobjectsok}{
  Let $N'$, $N$, $\tbtci$, $\tbsafetci$, $\tcidistance$, and $\tcivalidmarking$ be as
  in \refdf{tciimplobjects}.}  Then $N$ is a net as defined in Section~\ref{basic} and the
  clauses (1)--(5) of \reflem{branching-bisimulation} hold.
\end{lemmaG}
\begin{proof}
  Again, we use $\iprecond{x}$ and $\ipostcond{x}$ instead of
  $\precond{x}$ and $\postcond{x}$ when making assertions about the flow relation
  of $N$ (the implementation). Given that $\precond{t}\neq\emptyset$
  and $\precond{t}$ and $\postcond{t}$ are finite for all $t\inp T$
  and $\postcond{s}$ is finite for all $s \inp S$, by construction we
  have $\iprecond{t}\neq\emptyset$ and $\iprecond{t}$ and $\ipostcond{t}$ are
  finite for all $t\inp T \cup T^\tau$ and $\ipostcond{s}$ is finite
  for all $s \inp S \cup S^\tau$.  As $N$ has the same initial marking
  as $N'$, it must be finite. In order to show that $N$ is
  contact-free, we must show that for each reachable marking $M \in
  [M_0\rangle_N$ the following four properties are satisfied:
\begin{enumerate}
\vspace{-1ex}
\item[(i)] If $s \inp M$ then $s^{[t]} \notinp M$ for all $t \inp \postcond{s}$.
\vspace{-1ex}
\item[(ii)] If $s_t^{[u]},s^{[u]} \inp M$ then $\overline{s}_t^{[u]} \notinp M$.
\vspace{-1ex}
\item[(iii)] If $s^{[t]} \inp M$ for all $s \inp \precond{t}$ then
  $\circled{t} \notinp M$ and $s_t^{[u]} \notinp M$ for all $s \inp
  \precond{t}$ and $u \inp \postcond{s}$ with $[u]\neq[t]$.
\vspace{-1ex}
\item[(iv)] If $\circled{t}\inp M$ and $\overline{s}_t^{[u]} \inp M$ for all
  $s \inp \precond{t}$ and $u \inp \postcond{s}$ with $[u]\ne[t]$, then
  $M \cap \postcond{t} = \emptyset$.
\vspace{-1ex}
\end{enumerate}
  We proceed to show that all four properties are implied by $\tcivalidmarking(M)$.
  This entails that the contact-freeness of $N$ will follow
  immediately from the validity of clauses (1)--(3) of \reflem{branching-bisimulation}.

  Property (i) follows immediately from \ref{B2}$(M)$ and (ii) from \ref{B5}$(M)$.
  The claim $\circled{t}\notinp M$ of property~(iii) follows
  from \ref{B7}$(M)$, and using this the claim \plat{$s_t^{[u]} \notinp M$}
  from \ref{B5}$(M)$.  For (iv), assume, towards a contradiction, that
  $\circled{t}\inp M$, yet $s \inp M\cap \postcond{t}$. Then
  $\precond{t} \subseteq \tbsafetci(M)$. Now \ref{B1}$(M)$ and the
  contact-freeness of $N'$ gives
  $(\tbsafetci(M)\setminus \precond{t}) \cap \postcond{t}
  = \emptyset$. As $s \inp M\cap \postcond{t} \subseteq \tbsafetci(M)
  \cap \postcond{t}$ we obtain $s \inp \precond{t}$, contradicting \ref{B7}$(M)$.

  It remains to show the validity of clauses (1)--(5).
  Clause (1) follows directly from the definitions.

  \emph{Clause} (2): Assume $\tcivalidmarking(M_1)$. As remarked in
  Section~\ref{basic}, reachable markings of $N'$ are finite, so by
  \ref{B1}$(M_1\!)$ $M_1 \cap S$ is finite and $M_1$ contains only
  finitely many places of the form $\circled{t}$ (using \ref{B4}$(M_1)$
  and that $\precond{t}\ne\emptyset$ for $t\inp T$). Since for a given
  $t$, using that $\precond{t}$ and $\postcond{s}$ are finite, there
  are only finitely many places \plat{$s_t^{[u]}$} in $N$, it follows by
  \ref{B5}$(M_1)$ that $M_1$ contains only finite many places of the
  form \plat{$s_t^{[u]}$}.  From this we conclude that $\tcidistance(M_1)$ is
  finite.  We proceed by a case distinction over all transitions
  labelled $\tau$.

  Assume $M_1~[\boxed{s}\rangle_{N}~ M_2$. Then $M_2 = (M_1 \setminus \{s\}) \cup
  \{s^{[t]} \mid t \in \postcond{s}\}$ and $\tbtci(M_2) = \tbtci(M_1)$
  as well as $\tbsafetci(M_2) = \tbsafetci(M_1)$.
  Moreover, $\tcidistance(M_2) = \tcidistance(M_1) - 1$ as $s \inp M_1 \cap S$
  but $s \notinp M_2$ and the $s^{[t]}$ don't contribute to $\tcidistance$.
  It remains to check that $\tcivalidmarking(M_2)$. We will do that for
  each of the six conjuncts separately. 
  The validity of \ref{B1} is clearly preserved, in the sense that
  \ref{B1}$(M_1)$ implies \ref{B1}$(M_2)$. The same holds for
  \ref{B4} and \ref{B6}, as places  of the form $\circled{t}$ and
  \plat{$\overline{s}_t^{[u]}$} do not figure as pre- or postplaces of
  the transition $\boxed{s}$.
  Requirement \ref{B2}$\!^s(M_2)$ simply holds, as $s \notinp M_2$,
  whereas for $s' \ne s$ requirement \ref{B2}$\!^{s'}(M_2)$ is preserved.
  In the same way we obtain \ref{B3}$(M_2)$, \ref{B5}$(M_2)$ and \ref{B7}$(M_2)$.

  Assume $M_1~[t_s^{[u]}\rangle_{N}~ M_2$. Then $M_2 = (M_1 \setminus \{s_t^{[u]},
  s^{[u]}\}) \cup \{\overline{s}_t^{[u]}\}$. From $s_t^{[u]} \in M_1$ we obtain
  $\circled{t}\inp M_1$ by \ref{B5}$\!^{s,t,u}(M_1)$ and $s^{[t]}
  \notin M_1$ by \ref{B7}$(M_1)$. Hence the removal of any $s^{[u]}$
  does not affect $\tbtci$, and we have $\tbtci(M_2) = \tbtci(M_1)$.
  As the only change in summands contributing to $\tcidistance$ is the
  removal of \plat{$s_t^{[u]}$}, we have $\tcidistance(M_2) = \tcidistance(M_1) - 1$.
  Since $\circled{t} \inp M_1$, the removal of $s^{[u]}$ does not affect
  $\tbsafetci$ either, and we have $\tbsafetci(M_2) = \tbsafetci(M_1)$.
  Hence \ref{B1} is preserved.
  Requirement \ref{B2}$\!^{s,u}(M_2)$ holds (since $s^{[u]}\notinp M_2$)
  and \ref{B2}$\!^{s',t'}$ for $s'\ne s$ or $t'\ne u$ is preserved.
  Likewise, \ref{B3}$\!^{s,t',u}(M_2)$ holds (since $s^{[u]}\notinp M_2$)
  and \ref{B3}$\!^{s',t',u'}$ with $s'\ne s$ or $u' \ne u$ is preserved.
  Requirement \ref{B5}$\!^{s,t,u}(M_2)$ holds (because
  \plat{$s_t^{[u]}\notinp M_2$}), and \ref{B5}$\!^{s',t',u'}$ with $s'\ne s$
  or $u'\ne u$ is preserved. As for \ref{B5}$\!^{s,t',u}$ with $t'\ne
  t$, by \ref{B4}$(M_1)$ we have \plat{$\circled{t'}\notinp M_1$} and hence by
  \ref{B5}$\!^{s,t',u}(M_1)$ it must be that $s_{t'}^{[u]}\notinp
  M_1$, and thus $s_{t'}^{[u]}\notinp M_2$. This yields \ref{B5}$\!^{s,t',u}(M_2)$.
  Since $\circled{t}\inp M_1$ we have $\circled{t}\inp M_2$ and hence
  \ref{B6}$\!^{s,t,u}(M_2)$ holds. All other instances of \ref{B6} are preserved.
  Requirements \ref{B4} and \ref{B7} are preserved as well.

  Assume $M_1~[t'\rangle_{N}~ M_2$. Then $M_2 = (M_1 \setminus \{\circled{t},
  \overline{s}_t^{[u]} \mid s \inp \precond{t},~ u \inp \postcond{s},~ [u] \ne
  [t]\}) \cup \{s \mid s \inp \postcond{t}\}$ and $\tbtci(M_2) = \tbtci(M_1)$.
  Again $\tcidistance(M_2) = \tcidistance(M_1) -
  1$ as the single $\circled{t}$ contributed $1 + |\postcond{t}|$ whereas all
  the newly produced places $s$ together contribute $|\postcond{t}|$.
  As $\circled{t} \inp M_1$ we have $\precond{t} \in \tbsafetci(M_1)$.
  Moreover, for $s \inp \precond{t}$ and $u,v\inp\postcond{s}$, $[u]\ne
  [t]$, $v \ne t$ we have \plat{$\circled{t}, \overline{s}_t^{[u]} \inp M_1$}, so
  $s,s^{[t]},\circled{v}\notinp M_1$ by \ref{B7}($M_1$) and \ref{B4}($M_1$)
  and \plat{$s^{[u]}_t,s^{[u]}_v,s^{[u]} \notinp M_1$} by
  \ref{B5}($M_1$) and \ref{B3}($M_1$).
  Hence $\circled{t}$ is the only place in $M_1$ that contributes
  $s\inp \precond{t}$ to $\tbsafetci(M_1)$. Therefore $\tbsafetci(M_2) =
  (\tbsafetci(M_1) \setminus \precond{t}) \cup \postcond{t}$.
  Hence $\tbsafetci(M_1) ~[\{t\}\rangle_{N'}~ \tbsafetci(M_2)$,
  so \ref{B1} is preserved. Requirements \ref{B3}, \ref{B4}, \ref{B5} and \ref{B6}
  are easily seen to be preserved as well. 
  Since $N'$ is contact-free, we have $(\tbsafetci(M_1) \setminus
  \precond{t}) \cap \postcond{t} = \emptyset$, using \ref{B1}$(M_1)$.
  So for $s \inp \postcond{t}$ we have either $s \notinp
  \tbsafetci(M_1)$ or $s \inp \precond{t}$.  Either possibility implies
  $s^{[u]}\notinp M_1$ for $u \inp \postcond{s}$, and
  $\circled{v}\notinp M_1$ for $v \inp \postcond{s}$, $v\ne t$.
  Hence $s^{[u]},\circled{u}\notinp M_2$ for $u \inp \postcond{s}$.
  Using this, also \ref{B2} and \ref{B7} turn out to be preserved.

  \emph{Clause} (3):
  Assume $\tcivalidmarking(M_1) \wedge M_1 \,[G\rangle_{N}\, M_2$ with
  $\ell(t)\neq\tau$ for all $t\in G$.
  Then\vspace{-1ex} $$M_2 = (M_1\setminus \iprecond{G}) \cup \ipostcond{G}
  = M_1 \setminus \{s^{[t]} \mid s \inp \precond{G}\} \cup
  \{\circled{t}, s_t^{[u]} \mid t \inp G,~ s \inp \precond{t},~ u \in
  \postcond{s},~ [u] \ne [t]\}.$$
  For all $t \inp G$ and $s \inp \precond{t}$ we have $s^{[t]}\in M_1$ and hence
  $s \inp \tbsafetci(M_1)$. Thus $\tbsafetci(M_1) [t\rangle_{N'}$.

  \emph{Claim} 1: Let $t \inp G$, $s \inp \precond{t}$ and
  $u,v\inp\postcond{s}$. Then $\circled{v}\notinp M_1$ and
  $s^{[u]}\inp M_1$.

  \emph{Proof}: Assume, towards a contradiction, that $\circled{v}\in
  M_1$. Then $\precond{v} \subseteq \tbsafetci(M_1)$ and thus
  $\tbsafetci(M_1) [v\rangle_{N'}$.  As $s \in \precond{t}\cap\precond{v}$
  we have $\neg\tbsafetci(M_1) [t,v\rangle_{N'}$, so \ref{B1}$(M_1)$
  and \refdf{conflconcur} yield $t \# v$, and hence $[t]=[v]$.
  Nevertheless, \ref{B7}$(M_1)$ gives $s^{[v]}\notinp M_1$, whereas
  $s^{[t]}\inp M_1$.\contradiction

  Next assume that $s^{[u]}\notinp M_1$.
  Then \ref{B3}$\!^{s,u,t}(M_1)$ yields $\exists v \inp
  \postcond{s}. s_v^{[t]} \inp M_1$, and \ref{B5}$\!^{s,v,t}(M_1)$
  gives $\circled{v} \inp M_1$.\contradiction

  \emph{Claim} 2: Let $t_1, t_2 \inp G$ with $t_1 \neq t_2$. Then
  $\precond{t_1} \cap \precond{t_1} = \emptyset$.

  \emph{Proof}: Assume, towards a contradiction, that
  $s \inp \precond{t_1}\cap\precond{t_2}$.
  Then $\tbsafetci(M_1) [t_1\rangle_{N'}$ and
  $\tbsafetci(M_1) [t_2\rangle_{N'}$, but $\neg\tbsafetci(M_1)
  [t_1,t_2\rangle_{N'}$, so \ref{B1}$(M_1)$ and \refdf{conflconcur} yield
  $t_1 \# t_2$, and hence $[t_1]=[t_2]$. But this implies
  $s^{[t_1]} = s^{[t_2]} \in \iprecond{t_1}\cap\iprecond{t_2}$, contradicting
  $M_1 \,[G\rangle_{N}$.\contradiction

  \emph{Claim} 3: Let $t \inp G$, $s \inp \precond{t}$ and
  $v\inp\precond{s}$. Then $s,\circled{v}\notinp M_1$.

  \emph{Proof}: Since $s^{[t]} \inp M_1$ we have $s \notinp M_1$ by
  \ref{B2}$(M_1)$. Assume, towards a contradiction, that $\circled{v}\in
  M_1$. Then $\precond{v} \subseteq \tbsafetci(M_1) \in
  [M_0\rangle_{N'}$, using \ref{B1}$(M_1)$. As $N'$
  is contact-free, we have $(\tbsafetci(M_1) \setminus \precond{v}) \cap
  \postcond{v} = \emptyset$.  So since $s \inp \tbsafetci(M_1) \cap
  \postcond{v}$ it must be that $s \inp \precond{v}$.
  But then $\circled{v} \notinp M_1$ by Claim 1.\contradiction

  Claim 1 implies that $\precond{G} \subseteq \tbtci(M_1)$, and
  Claim 2 yields $\tbtci(M_1) ~[G\rangle_{N'}~M'_2$ for some $M'_2$.
  By Claim 3 we have $\tbtci(M_1\setminus \iprecond{G}) =
  \tbtci(M_1)\setminus \precond{G}$ and thus
  $$\tbtci(M_2) = \tbtci((M_1\setminus \iprecond{G})\cup \ipostcond{G})
  = (\tbtci(M_1)\setminus \precond{G}) \cup \postcond{G} = M'_2.$$
  It remains to check that $\tcivalidmarking(M_2)$.
  First of all, $\tbsafetci(M_2) = \tbsafetci(M_1)$ and hence \ref{B1}
  is preserved. It is easy to see that \ref{B2}, \ref{B6} and
  \ref{B7} are preserved. Requirement \ref{B3}$\!^{s}$ for $s \notinp
  \precond{G}$ is also preserved, whereas \ref{B3}$\!^{s}(M_2)$ for $s \inp
  \precond{t}$, $t \inp G$ holds with $v:=t$. Requirement \ref{B4}
  may fail to be preserved only if $\origexists t_1,t_2 \inp G$ with 
  $t_1 \ne t_2$ and $\precond{t_1}\cap\precond{t_2} \neq \emptyset$ or
  if $\origexists t \inp G$ and $\circled{v}\in M_1$ with 
  $\precond{t}\cap\precond{v} \neq \emptyset$.
  These cases are ruled out by Claims 2 and 1.
  Requirement \ref{B5}$\!^s$ with $s \notinp\precond{G}$ is preserved.
  Since there is no $\circled{v}\in M_1$ with
  $\precond{G}\cap\precond{v} \neq \emptyset$, by \ref{B5}$(M_1)$ and \ref{B6}$(M_1)$
  there are no \plat{$s_v^{[u]}, \overline{s}_v^{[u]} \in M_1$} with $s \inp \precond{G}$.
  Moreover, for all $t\inp G$, $s \inp \precond{t}$ and $u\inp
  \postcond{s}$ with $[u]\ne[t]$ we have $s^{[u]} \inp M_1$ and hence $s^{[u]} \inp M_2$.
  Thus we obtain \ref{B5}$\!^s(M_2)$ for $s \inp \precond{G}$.

  \emph{Clause} (4): By a case distinction on the three summands of $\tcidistance(M_1)$.

  Assume $\origexists s \inp M_1 \cap S$. Then $M_1 [\boxed{s}\rangle_{N'}$.

  Assume \plat{$\origexists s_t^{[u]} \inp M_1$}. Then by \ref{B5}$(M_1)$ also
  $s^{[u]} \in M_1$ and hence $M_1 [t_s^{[u]}\rangle_{N'}$.

  Assume $\origexists \circled{t} \inp M_1$ but
  \plat{$\neg\origexists s_t^{[u]} \inp M_1$}.
  Then by \ref{B7}$(M_1)$ also \plat{$\origexists \overline{s}_t^{[u]} \inp M_1$}
  for all $s \inp \precond{t}$ and $u\in \postcond{s}$ with $[u]\neq[t]$.
  Thus $M_1 [t'\rangle_{N'}$.

  \emph{Clause} (5): $\tcidistance(M_1) = 0$ implies $M_1 \cap S = \emptyset$ and
  $M_1$ does not contain places of the form $\circled{t}$ or $s_t^{[u]}$.
  By \ref{B6}$(M_1)$ it doesn't contain places of the form \plat{$\overline{s}_t^{[u]}$} either.
  Hence all places in $M_1$ have the form \plat{$s^{[t]}$} for $s\inp S$ and
  $t \inp \postcond{s}$. Moreover, by \ref{B3}$(M_1)$, for any $s \inp S$ either $M_1$
  contains all places $s^{[t]}$ with $t \inp \postcond{s}$ or none.
  Thus \plat{$M_1 = \{s^{[t]} \mid s \inp \tbtci(M_1),~ t \inp \postcond{s}\}$}.
  Using this, when $\tbtci(M_1) ~[G\rangle_{N'}~ M'_2$ for $G
  \subseteq T$, there is a unique $M_2$ such that \plat{$M_1 \,[G\rangle_N\, M_2$}.
  It remains to show that $\tbtci(M_2) = M'_2$.

  First of all, note that $M_2 \cap S = \emptyset$.
  Secondly, we have $$\{s \mid s \in S,~ \{s^{[t]} \mid t \in \postcond{s}\}
  \subseteq M_2\} = \{s \mid s \in \tbtci(M_1),~ s
  \not\in\precond{G}\} = \tbtci(M_1) \setminus \precond{G}.$$
  Finally, $\{s \mid s \in \postcond{t} \wedge \circled{t} \in M_2\} =
  \{s \mid s \in \postcond{t} \wedge t \in G\} = \postcond{G}$.

  Thus, applying Definitions~\ref{df-tciimplobjects}
  and~\ref{df-steps}, $\tbtci(M_2) = (\tbtci(M_1) \setminus
  \precond{G}) \cup  \postcond{G} = M'_2$.
\end{proof}

\begin{definition}{abstraction}{For $N$ a net and $i$ and action, let
  $N/i$ be the net obtained by renaming all occurrences of $i$ into $\tau$.}
\end{definition}

\begin{proposition}{pr-compositionality}
If $N \approx_\mathscr{R} N'$ then $N/i \approx_\mathscr{R} N'/i$.
\end{proposition}

\begin{proof}
$\rpair{\sigma,X}$ is a step ready pair of $N/i$ iff $N$ has a step
  ready pair $\rpair{\rho,X}$, where the sequence $\sigma$ can be
  obtained from $\rho$ by deleting all $i$'s, and $\{i\}\notinp X$.
\end{proof}

\begin{theorem}{tci}
Any net is step readiness equivalent to its
transition-controlled-choice implementation.
\end{theorem}

\begin{proof}
  Let $N'_\tau = (S, T, F', M_0, \ell'_\tau)$ be a net and $N_\tau = (S \cup S^\tau, T
  \cup T^\tau, F, M_0, \ell_\tau)$ its transition-controlled-choice implementation.
  Obtain $N'$ from $N'_\tau$ and $N$ from $N_\tau$ by changing all
  $\tau$-labels of transitions in $T$---but not those in
  $T^\tau$---into $i$. Thus $N=(S \cup S^\tau, T
  \cup T^\tau, F, M_0, \ell)$ where $\ell$ satisfies $\ell(t)=\tau$ if
  $t\inp T^\tau$; $\ell(t)=i$ if $t \inp T$ and $\ell_\tau(t)=\tau$; and
  $\ell(t)=\ell_\tau(t)$ otherwise. Then $N$ is still the
  transition-controlled-choice implementation of $N'$, and moreover
  $N'$ has no $\tau$-labels. Furthermore, $N'/i=N'_\tau$ and $N/i =N_\tau$.
  Lemmas~\ref{lem-branching-bisimulation} and~\ref{lem-tciimplobjectsok}
  yield $N \approx_\mathscr{R} N'$.
  So by \refpr{compositionality} we obtain $N/i \approx_\mathscr{R} N'/i$,
  which is $N_\tau \approx_\mathscr{R} N'_\tau$.
\end{proof}

\begin{thebibliography}{10}

\bibitem{aalst98beyond}
{ {W.M.P. van der} Aalst, E.~Kindler \& J.~Desel} (1998):
\newblock {\em Beyond asymmetric choice: A note on some extensions.}
\newblock {\sl Petri Net Newsletter} 55, pp. 3--13.

\bibitem{bes87}
{ E.~Best} (1987):
\newblock {\em {Structure theory of Petri nets: The free choice hiatus}.}
\newblock In W.~Brauer, W.~Reisig \& G.~Rozenberg, editors: {\sl Advances in Petri Nets 1986}, {\sl \rm LNCS} 254, Springer, pp. 168--206.

\bibitem{best83freesimple}
{ E.~Best \& M.W. Shields} (1983):
\newblock {\em Some equivalence results for free choice nets and simple nets and on the periodicity of live free choice nets.}
\newblock In G.~Ausiello \& M.~Protasi, editors: {\sl {\rm Proceedings 8th Colloquium on} Trees in Algebra and Programming (CAAP '83)}, {\sl \rm LNCS} 159, Springer, pp. 141--154.

\bibitem{boer91embedding}
{ F.S.~de Boer \& C.~Palamidessi} (1991):
\newblock {\em Embedding as a tool for language comparison: On the {CSP} hierarchy.}
\newblock In J.C.M. Baeten \& J.F. Groote, editors: {\sl {\rm Proceedings 2nd International Conference on} Concurrency Theory {\rm (CONCUR'91), Amsterdam, The Netherlands}}, {\sl \rm LNCS} 527, Springer, pp. 127--141.

\bibitem{bouge88symmetricleader}
{ L.~Boug\'{e}} (1988):
\newblock {\em On the existence of symmetric algorithms to find leaders in networks of communicating sequential processes.}
\newblock {\sl Acta Informatica} 25(2), pp. 179--201.

\bibitem{glabbeek08symmasymm}
{ R.J.~van Glabbeek, U.~Goltz \& J.-W. Schicke} (2008):
\newblock {\em Symmetric and asymmetric asynchronous interaction.}
\newblock Technical Report 2008-03, TU Braunschweig.
\newblock Extended abstract in Proceedings 1st {\sl Interaction and Concurrency Experience} (ICE'08) on {\sl Synchronous and Asynchronous Interactions in Concurrent Distributed Systems}, to appear in {\sl Electronic Notes in Theoretical Computer Science}, Elsevier.

\bibitem{GW96}
{ R.J.~van Glabbeek \& W.P. Weijland} (1996):
\newblock {\em Branching time and abstraction in bisimulation semantics.}
\newblock {\sl Journal of the ACM} 43(3), pp. 555--600.

\bibitem{G:FoSSaCS06}
{ D.~Gorla} (2006):
\newblock {\em On the relative expressive power of asynchronous communication primitives.}
\newblock In L.~Aceto \& A.~Ing{\'o}lfsd{\'o}ttir, editors: {\sl Proceedings 9th International Conference on Foundations of Software Science and Computation Structures (FoSSaCS '06)}, {\sl \rm LNCS} 3921, Springer, pp. 47--62.

\bibitem{hopkins91distnets}
{ R.P. Hopkins} (1991):
\newblock {\em Distributable nets.}
\newblock In {\sl Advances in Petri Nets 1991}, {\sl \rm LNCS} 524, Springer, pp. 161--187.

\bibitem{lamport78ordering}
{ L.~Lamport} (1978):
\newblock {\em Time, clocks, and the ordering of events in a distributed system.}
\newblock {\sl Communications of the ACM} 21(7), pp. 558--565.

\bibitem{lamport02arbitration}
{ L.~Lamport} (2003):
\newblock {\em Arbitration-free synchronization.}
\newblock {\sl Distributed Computing} 16(2-3), pp. 219--237.

\bibitem{lynch96}
{ N.~Lynch} (1996):
\newblock {\em Distributed Algorithms}.
\newblock Morgan Kaufmann Publishers.

\bibitem{nestmann00what}
{ U.~Nestmann} (2000):
\newblock {\em What is a `good' encoding of guarded choice?}
\newblock {\sl Information and Computation} 156, pp. 287--319.

\bibitem{OH86}
{ E.-R. Olderog \& C.A.R. Hoare} (1986):
\newblock {\em Specification-oriented semantics for communicating processes.}
\newblock {\sl Acta Informatica} 23, pp. 9--66.

\bibitem{palamidessi97comparing}
{ C.~Palamidessi} (1997):
\newblock {\em Comparing the expressive power of the synchronous and the asynchronous pi-calculus.}
\newblock In {\sl {\rm Conference Record of the 24th ACM SIGPLAN-SIGACT Symposium on} Principles of Programming Languages (POPL '97)}, ACM Press, pp. 256--265.

\bibitem{reisig82buffersync}
{ W.~Reisig} (1982):
\newblock {\em Deterministic buffer synchronization of sequential processes.}
\newblock {\sl Acta Informatica} 18, pp. 115--134.

\bibitem{selinger97firstorder}
{ Peter Selinger} (1997):
\newblock {\em First-order axioms for asynchrony.}
\newblock In {\sl {\rm Proceedings 8th International Conference on} Concurrency Theory {\rm (CONCUR'97), Warsaw, Poland}}, {\sl LNCS} 1243, Springer, pp. 376--390.

\bibitem{taubner88zurverteiltenimpl}
{ Dirk Taubner} (1988):
\newblock {\em Zur verteilten {I}mplementierung von {P}etrinetzen.}
\newblock {\sl Informationstechnik} 30(5), pp. 357--370.
\newblock Technical report, TUM-I 8805, TU {M\"unchen}.

\end{thebibliography}
\end{document}